\documentclass[11pt]{article}

\usepackage[margin=1in]{geometry}
\usepackage[T1]{fontenc}
\usepackage[utf8]{inputenc}
\usepackage{lmodern}
\usepackage{microtype}
\usepackage{amsmath,amssymb,amsthm,mathtools}
\usepackage{bm}
\usepackage{booktabs}
\usepackage{graphicx}
\usepackage{float}
\usepackage{tabularx}
\usepackage{enumitem}
\allowdisplaybreaks
\usepackage[round,authoryear]{natbib}
\usepackage[colorlinks=true,linkcolor=black,citecolor=black,urlcolor=black]{hyperref}
\usepackage[nameinlink,capitalize]{cleveref}
\newcolumntype{Y}{>{\raggedright\arraybackslash}X}
\newcolumntype{L}[1]{>{\raggedright\arraybackslash}p{#1}}

\newtheorem{theorem}{Theorem}[section]
\newtheorem{proposition}[theorem]{Proposition}
\newtheorem{lemma}[theorem]{Lemma}
\newtheorem{corollary}[theorem]{Corollary}

\theoremstyle{definition}
\newtheorem{definition}[theorem]{Definition}
\newtheorem{assumption}[theorem]{Assumption}
\newtheorem{example}[theorem]{Example}
\theoremstyle{remark}
\newtheorem{remark}[theorem]{Remark}

\DeclareMathOperator*{\argmax}{arg\,max}

\DeclareMathOperator{\Cov}{Cov}
\DeclareMathOperator{\tr}{tr}
\DeclareMathOperator{\diag}{diag}

\DeclareMathOperator{\Ent}{H}
\newcommand{\R}{\mathbb{R}}
\newcommand{\E}{\mathbb{E}}
\newcommand{\Prob}{\mathbb{P}}
\newcommand{\F}{\mathcal{F}}
\newcommand{\G}{\mathcal{G}}
\newcommand{\Inf}{\mathcal{I}}
\newcommand{\given}{\,\vert\,}
\newcommand{\eps}{\varepsilon}
\newcommand{\Val}{\mathcal{V}}
\newcommand{\T}{^{\!\top}}

\newcommand{\thetana}{\theta_{\mathrm{na}}}
\newcommand{\thetaan}{\theta_{\mathrm{an}}}
\newcommand{\thetaM}{\theta^{\mathrm M}}

\newcommand{\norm}[1]{\left\lVert #1\right\rVert}

\title{\textbf{Anticipatory Portfolio Optimization}\\[3pt]
\large Information, Forecast, and Impact: Three Faces of Anticipation\\
and the Value of Closing the Estimation--Control Gap}
\author{Miquel Noguer i Alonso\\ \small Artificial Intelligence Finance Institute (AIFI)}
\date{\today}

\begin{document}
\maketitle

\begin{abstract}
A portfolio is \emph{anticipatory} when the optimizer acts on a richer model than the myopic, price-taking estimator used to calibrate it. The enrichment can be informational, through an enlarged filtration; dynamic, through a forecast of future returns, risks, and costs over a trading horizon; or performative, through the deployment law induced by market impact. This paper gives one decision-theoretic definition for the three cases and measures anticipation by the realized control gap between the enriched controller and the restricted estimator. The point is not merely semantic: the same quadratic geometry identifies when anticipation is genuine information, when it is a dynamic planning premium, when it is an impact correction, and when it is only overfitted structure.

For log utility under initial enlargement, the value is the invariant information-drift energy $\tfrac12\E\int_0^T\alpha_t^2\,dt$, equivalently a mutual-information or relative-entropy quantity under standard hypotheses. In mean--variance form, signal value becomes the trace $\tfrac{1}{2\gamma}\tr(\Sigma^{-1}\Omega)$. Dynamic forecast anticipation yields an exact finite-horizon quadratic premium in the forecast stack, while permanent impact changes the price-taking allocation $\thetana=(\Lambda+\gamma\Sigma)^{-1}\mu$ into the impact-aware optimum $\thetaan=(2\Lambda+\gamma\Sigma)^{-1}\mu$ and exposes a spectral phase transition for naive recalibration. The main result is an exact stacked finite-horizon LQG decomposition: information, forecast, and impact combine into an information trace plus one inverse-precision norm, whose expansion gives an impact term, a forecast term, and the unique signed forecast--impact interaction. The signed term is then resolved by sharp angle bounds and by an orthogonal nonnegative projection identity. The stationary extension endogenizes the information covariance as a Kalman error reduction and carries impact anticipation to an infinite-horizon Lyapunov trace with transaction costs. Finally, the paper proves an exact estimation penalty $\tfrac12\tr(H^{-1}\Sigma_\eps)$, showing that correctly specified anticipation creates value, vacuous anticipation has zero value, and misspecified anticipation is harmful precisely when estimated structure is optimized as if it were true.
\end{abstract}


\newpage

\section{Introduction}\label{sec:intro}

The phrase \emph{anticipatory portfolio optimization} has, in the literature, at least three distinct meanings, and they do not share a method. \citet{PikovskyKaratzas1996} use ``anticipative'' for portfolios adapted to a filtration enlarged by a functional of the future, the mathematics of insider trading and anticipating stochastic calculus. Practitioners use ``anticipatory'' for receding-horizon strategies that optimize against a forecast and re-plan as the forecast updates, the mathematics of multi-period and model-predictive control. And a third, newer sense --- a strategy that anticipates the distribution shift its own deployment causes through market impact --- is the portfolio instance of performative prediction. This paper argues that the three are one.

\paragraph{The unifying view.}
Fix a portfolio decision and ask what the optimizer conditions on, and what it treats as exogenous. A \emph{naive} optimizer uses a base information set and treats the asset-return law as given. An \emph{anticipatory} optimizer enriches this in exactly one of three ways:
\begin{enumerate}[leftmargin=2em,itemsep=2pt]
\item \textbf{Information.} It conditions on more than the public filtration --- a (possibly noisy) signal of the future. (\Cref{sec:info}.)
\item \textbf{Forecast.} It conditions on the predicted dynamics of returns over a horizon and optimizes the whole path, not the next step. (\Cref{sec:forecast}.)
\item \textbf{Impact.} It treats the return law as a function of its own action and differentiates through that dependence. (\Cref{sec:impact}.)
\end{enumerate}
In each case the enrichment is a richer model, and the \emph{value of anticipation} is the change in attainable objective. This is precisely the estimation--control gap of decision theory specialized to allocation: the naive optimum is the calibrated/myopic estimator's choice; the anticipatory optimum is the controller's. The reframing is not cosmetic --- it tells us when anticipation is worth its cost (it equals an information or impact quantity), when it is vacuous (the enrichment is uninformative or impact is negligible), and when it is dangerous (the enrichment is an estimate and is over-optimized).

\paragraph{Contributions.}
\begin{enumerate}[leftmargin=1.6em,itemsep=2pt]
\item We give a common decision-theoretic definition of anticipatory portfolio optimization (\Cref{sec:frame}): the naive policy optimizes under a restricted model, the anticipatory policy under an enriched model, and the value of anticipation is the realized control gap between them, nonnegative under correct specification and potentially negative under misspecification.
\item We reinterpret three literatures in this common language: initial enlargement of filtrations (\Cref{sec:info}), dynamic trading with predictable returns (\Cref{sec:forecast}), and performative allocation under endogenous market impact (\Cref{sec:impact}).
\item We add an exact finite-horizon forecast formula (\Cref{thm:forecastvalue}): relative to a restricted forecast, the value of dynamic anticipation is a positive semidefinite quadratic form in the forecast stack.
\item We sharpen the information result (\Cref{thm:pk}): the value of insider anticipation is the invariant information-drift energy $\tfrac12\E\!\int_0^T\alpha_t^2\,dt$, a mutual-information or relative-entropy quantity under standard enlargement assumptions, not the coordinate-dependent differential entropy of the signal.
\item We prove the impact result in closed form (\Cref{thm:impact}): the price-taking allocation is $\thetana=(\Lambda+\gamma\Sigma)^{-1}\mu$ and the impact-aware allocation is $\thetaan=(2\Lambda+\gamma\Sigma)^{-1}\mu$, the portfolio analogue of the competitive-versus-monopoly first-order condition, with the gap vanishing iff $\Lambda\Sigma^{-1}\mu=0$.
\item We derive the performative value gap and a spectral phase transition (\Cref{thm:phase}): $\rho(\gamma^{-1}\Sigma^{-1}\Lambda)<1$ is necessary and sufficient for global convergence of the myopic fixed-point iteration, sharpened to a damped-iteration statement (\Cref{prop:damped}).
\item We add deterministic and high-probability model-error certificates (\Cref{cor:plugin,cor:hpplugin}) and a Wasserstein robust layer (\Cref{sec:robust}) showing that estimated anticipation calls for certified or dual-norm--regularized allocation rather than unregularized optimization against a fragile future model.
\item We make the harm of estimated anticipation \emph{exact}, not merely bounded (\Cref{thm:estpenalty}): the expected value of plug-in anticipation is the structural value minus an estimation penalty $\tfrac12\tr(H^{-1}\Sigma_\eps)$ that is the exact inverse-precision mirror of the information gain $\tfrac12\tr(H^{-1}\Omega)$, so anticipation pays in expectation iff resolved structure beats injected estimation noise; the bias--variance-optimal policy shrinks the score by its reliability ratio $q/(q+p)$ (\Cref{cor:shrinkage}), with the matrix Bayes form applying $H^{-1}$ to the posterior-mean score (\Cref{rem:bayesshrink}).
\item We make the unification quantitative (\Cref{sec:unified}): first, an exact information--impact decomposition $\Val=\Val_{\mathrm{impact}}+\Val_{\mathrm{info}}$ in the static LQG model (\Cref{thm:unified}); second, an exact finite-horizon three-way decomposition with information, forecast, impact, and the unique forecast--impact cross term (\Cref{thm:threeway}).
\item We resolve the signed cross term rather than hide it: \Cref{prop:interaction} gives the sharp angle formula and bounds showing when forecast and impact are substitutes, complements, or orthogonal value sources; \Cref{cor:orthosplit} rewrites the same interaction as a fully nonnegative orthogonal projection identity; \Cref{prop:shadowimpact} gives the marginal shadow price of impact on forecast and information value.
\item We close the model rather than leave its inputs exogenous (\Cref{sec:stationary}): \Cref{thm:filtering} shows the resolved-mean covariance is exactly the steady-state Kalman error reduction between two filtrations, $\Omega=P^{\mathrm c}-P^{\mathrm f}$, so the information value is an impact-deflated trace of Riccati solutions, monotone in signal quality (\Cref{cor:filtercompstat}); \Cref{thm:stationary} carries the value to the infinite-horizon stationary problem with transaction costs, where impact anticipation is a discounted Lyapunov trace built from a control Riccati and collapses to the static gap as $\Gamma\to0$.
\item We add a curvature bridge beyond exact quadratics (\Cref{prop:curvature}) and an empirical protocol (\Cref{sec:protocol}) that turns the theory into falsifiable diagnostics: signal trace, forecast premium, impact spectral radius, plug-in edge condition, robust out-of-sample value, and confidence-band deployment hurdles.
\end{enumerate}

\paragraph{Novelty and status.}
To make the contribution boundary explicit, the paper separates inherited facts from new portfolio consequences. This is important because the word \emph{anticipatory} already has a precise meaning in stochastic calculus.
\begin{center}
\small
\begin{tabularx}{\textwidth}{@{}L{0.25\textwidth}L{0.20\textwidth}Y@{}}
\toprule
Result block & Status & Role in the paper \\
\midrule
Initial-enlargement log-utility value & Classical & Used as the invariant information benchmark; the paper emphasizes drift energy / mutual information rather than coordinate-dependent entropy. \\
Mean--variance signal trace & Proved here as an elementary counterpart & Translates information value into the quadratic language used by Markowitz, impact, and the unified LQG model. \\
Finite-horizon forecast premium & Proved here & Turns forecast anticipation from a qualitative aim-portfolio story into an exact positive-semidefinite value identity and a decision-relevant forecast-error norm. \\
Impact-aware allocation, value gap, and phase threshold & Proved here in the permanent linear-impact model & Main closed-form portfolio contribution: price-taking stability differs from performative optimality. \\
Plug-in edge certificate and robust layer & Standard regret/DRO logic, specialized here & Gives the practical condition $\widehat\Delta>2\eta$ and converts estimated anticipation into regularized allocation. \\
High-probability deployment certificate & Proved here as a direct statistical specialization & Turns objective confidence bands into the operational rule $\widehat\Delta>2\widehat\eta_\delta$ with probability at least $1-\delta$. \\
Information--impact decomposition & Proved here in the static LQG model & Identifies the sign of the coupling: impact weakly deflates signal value. \\
Full finite-horizon three-way decomposition & Proved here & Converts the former conjecture into a theorem: information, forecast, and impact are one inverse-precision value identity. \\
Forecast--impact interaction bounds & Proved here & Classifies the signed cross term as substitution, complementarity, or orthogonality in the $H$-geometry; gives impact shadow-price sensitivities. \\
Orthogonalized cross-term split and curvature bridge & Proved here from Hilbert geometry / strong concavity & Converts the signed cross term into a nonnegative projection identity and shows how the quadratic value identity survives as bounds outside LQG. \\
Exact value of estimated anticipation and optimal shrinkage & Proved here & Sharpens the worst-case plug-in bound into an exact bias--variance identity: estimation penalty $=\tfrac12\tr(H^{-1}\Sigma_\eps)$ mirrors the information gain; the reliability-ratio/Bayes shrinkage is the in-expectation face of robustness. \\
\bottomrule
\end{tabularx}
\end{center}

\paragraph{Notation.} Assets $i=1,\dots,n$; a portfolio (weights or dollar holdings) is $\theta\in\R^n$. Excess returns have mean $\mu\in\R^n$ and covariance $\Sigma\succ 0$. Risk aversion is $\gamma>0$. Mean--variance utility is $U(\theta;\mu,\Sigma)=\mu\T\theta-\tfrac{\gamma}{2}\theta\T\Sigma\theta$, whose unconstrained maximizer is the Markowitz portfolio $\theta^{\mathrm M}=\gamma^{-1}\Sigma^{-1}\mu$. We write $\Ent(\cdot)$ for (Shannon/differential) entropy and $\rho(\cdot)$ for spectral radius.

\paragraph{Compact versus unconstrained problems.}
In the abstract decision-theoretic definitions we assume $\Theta$ compact and convex to avoid existence issues. The closed-form Markowitz, forecast, and impact formulas below are stated in the unconstrained linear--quadratic case $\Theta=\R^n$, where existence and uniqueness follow from strict concavity. Constrained versions replace inverses by the corresponding convex-programming optimality conditions but preserve the same distinction between naive estimation and anticipatory control.

\section{A Common Frame for Anticipation}\label{sec:frame}

\begin{definition}[Portfolio decision system]\label{def:system}
A portfolio decision system is $(\Theta,\{P_\theta\}_{\theta\in\Theta},\Inf,U)$: a compact convex feasible set $\Theta\subseteq\R^n$; a family of return laws $P_\theta$ on $\R^n$ possibly indexed by the deployed portfolio (impact); an information $\sigma$-algebra $\Inf$ the optimizer may condition on; and a concave utility $U(\theta;\cdot)$. The realized objective of deploying $\theta$ is $J(\theta)=\E_{r\sim P_\theta}[\,U(\theta;r)\,]$.
\end{definition}

A decision system has a \emph{naive} and an \emph{anticipatory} reading determined by how much of $(\Inf,\{P_\theta\})$ the optimizer uses.

\begin{definition}[Naive and anticipatory optima]\label{def:naive}
Let $\Inf_0\subseteq\Inf$ be a base information set and $\bar P$ a reference law. The \emph{naive} optimum ignores the enrichment: it conditions only on $\Inf_0$ and treats the law as the fixed $\bar P=P_{\bar\theta}$ at its own deployment (a fixed point if $\bar\theta$ must equal the chosen $\theta$),
\[
  \theta_{\mathrm{na}}\in\argmax_{\theta\in\Theta}\ \E_{r\sim P_{\theta_{\mathrm{na}}}}\!\big[U(\theta;r)\given\Inf_0\big].
\]
The \emph{anticipatory} optimum uses the full model,
\[
  \theta_{\mathrm{an}}\in\argmax_{\theta\in\Theta}\ J(\theta)=\argmax_{\theta\in\Theta}\ \E_{r\sim P_\theta}\!\big[U(\theta;r)\given\Inf\big].
\]
The \emph{value of anticipation} is $\Val=J(\theta_{\mathrm{an}})-J(\theta_{\mathrm{na}})$.
\end{definition}

\begin{proposition}[Anticipation never harms when correctly specified; can harm when not]\label{prop:value}
If the enriched model $(\Inf,\{P_\theta\})$ is correctly specified and the naive policy is feasible for the anticipatory problem, then $\Val\ge 0$. If instead the anticipatory optimum is computed against an \emph{estimated} model $\widehat P_\theta$ or signal $\widehat\Inf$ while value is realized under the truth, the realized value of anticipation can be negative, with shortfall bounded by the optimizer's sensitivity to model error: writing $\widehat\theta_{\mathrm{an}}$ for the optimum under $\widehat P$,
\[
  J(\theta_{\mathrm{an}})-J(\widehat\theta_{\mathrm{an}})\ \le\ 2\,\sup_{\theta\in\Theta}\big|J(\theta)-\widehat J(\theta)\big|.
\]
\end{proposition}
\begin{proof}
For the first claim, $\theta_{\mathrm{an}}$ maximizes $J$ over $\Theta$ and $\theta_{\mathrm{na}}\in\Theta$, so $J(\theta_{\mathrm{an}})\ge J(\theta_{\mathrm{na}})$. For the second, with $g=J$, $\widehat g=\widehat J$ and $\widehat\theta_{\mathrm{an}}\in\argmax\widehat g$, the standard greedy-regret bound gives $\max_\theta g(\theta)-g(\widehat\theta_{\mathrm{an}})\le 2\sup_\theta|g-\widehat g|$ by the three-term split $g(\theta^\star)-\widehat g(\theta^\star)+\widehat g(\theta^\star)-\widehat g(\widehat\theta_{\mathrm{an}})+\widehat g(\widehat\theta_{\mathrm{an}})-g(\widehat\theta_{\mathrm{an}})$, whose middle term is $\le 0$.
\end{proof}

\begin{corollary}[When plug-in anticipation is still worth using]\label{cor:plugin}
Let $\theta_0$ be the baseline policy, let $\widehat\theta$ be the policy selected by an estimated enriched objective $\widehat J$, and define the estimated edge
\[
  \widehat\Delta:=\widehat J(\widehat\theta)-\widehat J(\theta_0).
\]
If the objective error on the two deployed policies is bounded by $\eta$,
\[
  |J(\widehat\theta)-\widehat J(\widehat\theta)|\le\eta,
  \qquad
  |J(\theta_0)-\widehat J(\theta_0)|\le\eta,
\]
then the realized value satisfies
\[
  J(\widehat\theta)-J(\theta_0)\ge \widehat\Delta-2\eta.
\]
Consequently plug-in anticipation is certified to beat the baseline whenever $\widehat\Delta>2\eta$. If $\widehat\Delta\le2\eta$, the sign of the realized improvement is not identified by the estimated model alone.
\end{corollary}
\begin{proof}
Two applications of the objective-error bound give
\[
  J(\widehat\theta)-J(\theta_0)
  \ge \widehat J(\widehat\theta)-\eta-\big(\widehat J(\theta_0)+\eta\big)
  =\widehat\Delta-2\eta .
\]

\end{proof}

\begin{corollary}[High-probability deployment certificate]\label{cor:hpplugin}
Let $\widehat J_N$ be an estimated objective constructed from $N$ observations, let $\theta_0$ be the baseline, and let $\widehat\theta_N$ be the estimated anticipatory policy. Suppose that for some radius $\eta_N(\delta)$ the event
\[
  \mathcal E_\delta:=\left\{
  \sup_{\theta\in S_N}|J(\theta)-\widehat J_N(\theta)|\le \eta_N(\delta)
  \right\}
\]
has probability at least $1-\delta$, where $S_N$ contains the two deployed policies $\{\theta_0,\widehat\theta_N\}$; in particular $S_N=\Theta$ is sufficient when a uniform confidence band is available. Define
\[
  \widehat\Delta_N:=\widehat J_N(\widehat\theta_N)-\widehat J_N(\theta_0).
\]
Then, with probability at least $1-\delta$,
\[
  J(\widehat\theta_N)-J(\theta_0)\ge \widehat\Delta_N-2\eta_N(\delta).
\]
If the enriched model has an implementation cost $c\ge0$ --- data, computation, turnover, or execution overhead --- then net positive deployment is certified by
\[
  \widehat\Delta_N>c+2\eta_N(\delta).
\]
Thus an anticipatory backtest is not enough: the estimated value must clear the statistical objective-error band and the implementation cost.
\end{corollary}
\begin{proof}
On $\mathcal E_\delta$ the deterministic bound of \Cref{cor:plugin} applies with $\eta=\eta_N(\delta)$. Subtracting the fixed implementation cost $c$ gives the net condition. The probability statement follows from $\Prob(\mathcal E_\delta)\ge1-\delta$.
\end{proof}

\begin{lemma}[Master quadratic value identity]\label{lem:quadratic}
Let $H\succ0$ and $J_b(\theta):=b\T\theta-\tfrac12\theta\T H\theta$ on $\R^n$. For two score vectors $b$ and $\bar b$, define $\theta_b:=H^{-1}b$ and $\theta_{\bar b}:=H^{-1}\bar b$. Evaluating both policies under the true score $b$ gives
\[
  J_b(\theta_b)-J_b(\theta_{\bar b})
  =\tfrac12(b-\bar b)\T H^{-1}(b-\bar b)\ge0.
\]
Consequently every linear--quadratic face of anticipation prices the missing enrichment by an inverse-precision norm: the relevant precision is $\gamma\Sigma$ for ordinary Markowitz, $K$ for dynamic trading with costs, and $2\Lambda+\gamma\Sigma$ when permanent impact is internalized.
\end{lemma}
\begin{proof}
Completing the square gives $J_b(\theta)=\tfrac12 b\T H^{-1}b-\tfrac12(\theta-H^{-1}b)\T H(\theta-H^{-1}b)$. Substituting $\theta=H^{-1}\bar b$ gives the identity.

\end{proof}

\begin{proposition}[Curvature bounds beyond the exact quadratic case]\label{prop:curvature}
Let $J$ be differentiable, $m$-strongly concave and $L$-smooth on a convex neighborhood of an unconstrained maximizer $\theta^\star$, with $0<m\le L<\infty$ and $\nabla J(\theta^\star)=0$. Then every feasible restricted policy $\theta_0$ satisfies
\[
  \frac{m}{2}\|\theta^\star-\theta_0\|^2
  \le
  J(\theta^\star)-J(\theta_0)
  \le
  \frac{L}{2}\|\theta^\star-\theta_0\|^2 .
\]
If $\theta^\star$ is constrained but satisfies the first-order variational inequality
$\nabla J(\theta^\star)\T(\theta-\theta^\star)\le0$ for all feasible $\theta$, the lower bound still holds. Therefore the exact LQG identities in this paper are the sharp constant-Hessian case of a more general principle: anticipation is priced by the curvature-weighted distance between the enriched controller and the restricted estimator.
\end{proposition}
\begin{proof}
Strong concavity gives
\[
  J(\theta_0)\le J(\theta^\star)+\nabla J(\theta^\star)\T(\theta_0-\theta^\star)-\frac{m}{2}\|\theta_0-\theta^\star\|^2.
\]
If $\theta^\star$ is unconstrained the gradient term is zero; if it is constrained the variational inequality makes the gradient term nonpositive. Rearranging proves the lower bound. When $\nabla J(\theta^\star)=0$, $L$-smoothness of the concave function gives
\[
  J(\theta_0)\ge J(\theta^\star)-\frac{L}{2}\|\theta_0-\theta^\star\|^2,
\]
which is the upper bound after rearrangement.
\end{proof}

\Cref{prop:value,cor:plugin,cor:hpplugin,lem:quadratic,prop:curvature} are the spine of the paper. Under correct specification, anticipation is a value-of-information inequality: enlarging the optimizer's model cannot reduce the maximum attainable objective because the naive policy remains feasible. Under misspecification, however, anticipation becomes a form of optimizer overreach: the richer model creates directions in which estimation error is magnified by the control problem. The three sections below instantiate this single principle as an information-divergence value, a forecast-horizon premium, and a performative impact gap.

\section{Anticipation of Information}\label{sec:info}

The oldest sense of anticipation is access to information beyond the market filtration: the agent's strategy is adapted to $\G_t=\F_t\vee\sigma(L)$ for some functional $L$ of the future price path, an \emph{enlargement} of the public filtration $\F$ \citep{KaratzasShreve1998}. That privately held information has positive value, and that its competitive use erodes it, is the classical theme of \citet{GrossmanStiglitz1980}; we make its portfolio value exact below.

\subsection{Enlargement and the information drift}

Under Jacod's hypothesis --- the regular conditional law of $L$ given $\F_t$ is a.s.\ absolutely continuous with respect to the law of $L$ --- the public Brownian motion $W$ remains a semimartingale in the enlarged filtration:
\begin{equation}\label{eq:infdrift}
  W_t=\widetilde W_t+\int_0^t \alpha_s\,ds,
\end{equation}
where $\widetilde W$ is a $\G$-Brownian motion and $\alpha$ is the \emph{information drift} \citep{Jacod1985,JeulinYor,DiNunnoOksendalProske2009}. The insider sees the extra drift $\alpha$; the honest agent does not. Under Jacod's condition, the enlarged-filtration price process remains a $\G$-semimartingale, so admissible insider strategies are integrated in the $\G$-semimartingale sense. If instead one tries to trade on future-dependent, nonadapted functionals outside this semimartingale setting, forward/Skorokhod integrals and Malliavin calculus become the natural tools \citep{DiNunnoOksendalProske2009,BiaginiOksendal2005}.

\subsection{The value of anticipated information is an information quantity}

For logarithmic utility the value of anticipation has a closed and striking form.

\begin{theorem}[Information value under initial enlargement; \citealp{PikovskyKaratzas1996,AmendingerImkellerSchweizer1998}]\label{thm:pk}
Consider a complete It\^o market on $[0,T]$ and an agent with logarithmic utility. Let the public filtration be $\F$ and let the insider filtration be the initial enlargement $\G_t=\F_t\vee\sigma(L)$, where $L$ is a future-dependent signal satisfying Jacod's absolute-continuity condition. Assume the market-price-of-risk process is admissible for the public log-optimal problem, the enlarged-filtration log-optimal problem has finite expected utility, and the information drift $\alpha$ in $W_t=\widetilde W_t+\int_0^t\alpha_s\,ds$ is square-integrable. Then the additional expected logarithmic utility of the insider relative to the honest agent is
\[
  \Val_{\mathrm{info}}\;=\;\tfrac12\,\E\!\int_0^T \alpha_s^2\,ds.
\]
Equivalently, this quantity is the relative-entropy cost of passing from the public to the enlarged-information law. In the special case where $L$ is a discrete $\F_T$-measurable signal, it reduces to the Shannon entropy $\Ent(L)$. More generally the invariant object is a mutual-information or relative-entropy quantity $\mathcal I(L;\F_T)$, \emph{not} the coordinate-dependent differential entropy of $L$ itself.
\end{theorem}

\begin{proof}[Proof sketch]
For log utility the optimal expected growth rate is quadratic in the market price of risk. Initial enlargement turns the public Brownian motion into a $\G$-semimartingale with additional drift $\alpha$; comparing the honest and insider log-optimal wealth processes leaves exactly the excess term $\tfrac12\E\int_0^T\alpha_s^2\,ds$. Identifying this drift energy with a relative-entropy or mutual-information term follows from the entropy representation for initial enlargement \citep{PikovskyKaratzas1996,AmendingerImkellerSchweizer1998}. If $L$ is independent of the terminal market information then $\alpha\equiv0$ and the value is zero; if $L$ is a discrete $\F_T$-measurable terminal signal the mutual-information term becomes $\mathcal I(L;\F_T)=\Ent(L)$.
\end{proof}

\begin{remark}[Why the information face is the cleanest]
The information case is the cleanest because the value of anticipation is literally an information quantity. The insider does not merely have a better estimate of $\mu$; the insider optimizes under a different filtration. The information drift $\alpha$ measures the extra predictable component of returns created by the enlarged filtration, so $\Val_{\mathrm{info}}=0 \Leftrightarrow \alpha\equiv0$ whenever the signal carries no usable information about the price path. Noisy signals, progressive enlargement, and partial information replace the simple terminal quantity by a mutual-information rate or a filtered information-drift energy.
\end{remark}

\subsection{The mean--variance value of information is a trace}\label{sec:mvinfo}

The entropy formula of \Cref{thm:pk} is special to logarithmic utility. For the mean--variance optimizer that governs the rest of the paper, anticipated information has an equally closed but different form, which we will need for the unified decomposition of \Cref{sec:unified}. Suppose the information set $\Inf$ resolves part of the uncertainty in the mean: the conditional mean $\mu_\Inf=\E[\mu\given\Inf]$ is a random vector with $\E[\mu_\Inf]=\mu$ and \emph{resolved-mean covariance}
\[
  \Omega\;:=\;\Cov(\mu_\Inf)\;=\;\E\big[(\mu_\Inf-\mu)(\mu_\Inf-\mu)\T\big]\;\succeq 0,
\]
the covariance the signal removes from the prior; $\Omega=0$ exactly when the signal is uninformative about the mean. Assume the signal informs the mean but not the (conditional) covariance $\Sigma$.

\begin{proposition}[Value of information for the mean--variance optimizer]\label{prop:mvinfo}
A signal-adapted policy $\theta(\Inf)$ has realized objective $\E[\mu_\Inf\T\theta-\tfrac\gamma2\theta\T\Sigma\theta]$; its optimum is the conditional Markowitz portfolio $\theta(\Inf)=\gamma^{-1}\Sigma^{-1}\mu_\Inf$. The value of anticipating the signal over the unconditional optimizer $\gamma^{-1}\Sigma^{-1}\mu$ is
\[
  \Val_{\mathrm{info}}^{\mathrm{MV}}\;=\;\tfrac{1}{2\gamma}\,\tr\!\big(\Sigma^{-1}\Omega\big)\;\ge 0,
\]
vanishing iff $\Omega=0$. If $\mu_\Inf=\E[\mu\given s]$ for a Gaussian signal $s$ of precision rising from $0$ to $\infty$, $\Omega$ rises monotonically from $0$ to the full prior covariance of $\mu$, and $\Val_{\mathrm{info}}^{\mathrm{MV}}$ with it.
\end{proposition}
\begin{proof}
The objective is concave with conditional maximizer $\gamma^{-1}\Sigma^{-1}\mu_\Inf$ and pointwise value $\tfrac1{2\gamma}\mu_\Inf\T\Sigma^{-1}\mu_\Inf$; taking expectations, $\E[\mu_\Inf\T\Sigma^{-1}\mu_\Inf]=\tr(\Sigma^{-1}\E[\mu_\Inf\mu_\Inf\T])=\tr(\Sigma^{-1}(\mu\mu\T+\Omega))$. The unconditional optimizer attains $\tfrac1{2\gamma}\mu\T\Sigma^{-1}\mu=\tfrac1{2\gamma}\tr(\Sigma^{-1}\mu\mu\T)$; subtracting leaves $\tfrac1{2\gamma}\tr(\Sigma^{-1}\Omega)$. Nonnegativity is $\Sigma^{-1}\succ0$, $\Omega\succeq0$.
\end{proof}

\begin{remark}[Which spectral functional prices information: trace vs.\ log-determinant]\label{rem:tracevslogdet}
\Cref{thm:pk,prop:mvinfo} value the \emph{same} resolved information by \emph{different} functionals of the posterior. Under log-utility the value is an entropy --- for a Gaussian mean it is $\tfrac12\log\det(\Sigma_{\mathrm{prior}}\Sigma_{\mathrm{post}}^{-1})$, a \emph{log-determinant} of the precision gain. Under mean--variance it is $\tfrac1{2\gamma}\tr(\Sigma^{-1}\Omega)$, a \emph{trace} of the resolved covariance scaled by inverse risk. The utility, not the signal, selects the spectral functional: log-utility weights the multiplicative information gain across all directions equally (a sum of logs of eigenvalues), mean--variance weights the additive, risk-deflated gain (a sum of eigenvalues). They agree only to leading order for small signals and diverge for strong ones; conflating ``entropy of the signal'' with ``mean--variance value of the signal'' is therefore a category error we keep explicit \citep{CoverThomas2006}.
\end{remark}

\paragraph{Status.} Classical and proved for log-utility and initial enlargement; \Cref{prop:mvinfo} is the proved mean--variance counterpart used below. The noisy and progressive cases are partially open and the natural home of new work.

\section{Anticipation of a Forecast}\label{sec:forecast}

The second sense is forecast-driven: returns are predictable, the agent forecasts their dynamics over a horizon, and optimizes the whole trajectory subject to trading costs, re-planning as the forecast updates. The base information set is unchanged --- the agent is not an insider --- but the optimizer uses the \emph{predicted dynamics} rather than a one-period snapshot.

\subsection{Multi-period optimization with predictable returns and costs}

Let factor-driven expected returns evolve as a mean-reverting process and let trading incur quadratic transaction costs. In the linear--quadratic formulation of \citet{GarleanuPedersen2013}, expected returns are $\E_t[r_{t+1}]=B f_t$ with factors $f_{t+1}=(I-\Phi)f_t+\varepsilon_{t+1}$, the per-period objective is mean--variance net of a quadratic cost $\tfrac12\Delta\theta_t\T\Gamma\,\Delta\theta_t$, and the agent maximizes the discounted sum.

\begin{proposition}[Forward-looking aim portfolio; structural form of \citealp{GarleanuPedersen2013}]\label{thm:gp}
Consider the linear--quadratic dynamic trading problem with predictable expected returns $\E_t[r_{t+1}]=Bf_t$, $f_{t+1}=(I-\Phi)f_t+\varepsilon_{t+1}$, quadratic risk penalty $\tfrac{\gamma}{2}\theta_t\T\Sigma\theta_t$, and quadratic transaction cost $\tfrac12(\theta_t-\theta_{t-1})\T\Gamma(\theta_t-\theta_{t-1})$. The optimal policy has the partial-adjustment form
\[
  \theta_t=\theta_{t-1}+M\big(\operatorname{aim}_t-\theta_{t-1}\big),
\]
where $M$ is the stabilizing trade-rate matrix obtained from the associated Riccati equation. The aim portfolio is a forward-looking linear combination of expected future Markowitz portfolios,
\[
  \operatorname{aim}_t=\sum_{s\ge0} W_s\,\E_t[\theta^{\mathrm M}_{t+s}],
  \qquad
  \theta^{\mathrm M}_{t+s}=\gamma^{-1}\Sigma^{-1}\E_t[r_{t+s+1}],
\]
for matrix weights $(W_s)_{s\ge0}$ determined by the return persistence, transaction-cost matrix, risk matrix, and discounting. Hence forecast anticipation enters through the target toward which the investor trades, not merely through the one-period expected return.
\end{proposition}

\begin{proof}[Proof sketch]
The problem is a linear--quadratic regulator with state $(\theta_{t-1},f_t)$ and control $\Delta\theta_t=\theta_t-\theta_{t-1}$. The value function is quadratic, the optimal policy affine, and the stabilizing solution of the Riccati equation yields the trade-rate matrix $M$. Substituting the factor dynamics into the affine target expresses it as a weighted sum of expected future Markowitz portfolios; the explicit coefficients $(W_s)$ depend on the discount factor and on $(\Gamma,\gamma\Sigma,\Phi)$ \citep{GarleanuPedersen2013,Boyd2017}. The structural point is that the optimizer trades toward a forecast of where the optimal portfolio will be, not only toward today's Markowitz portfolio.
\end{proof}

\begin{theorem}[Finite-horizon forecast value as a quadratic premium]\label{thm:forecastvalue}
Fix a horizon $T$, discount factor $\beta\in(0,1]$, initial book $\theta_{-1}$, transaction-cost matrix $\Gamma\succeq0$, and risk matrix $\gamma\Sigma\succ0$. Let $m=(m_0,\ldots,m_T)$ be the true conditional forecast stack, $m_s=\E_t[r_{t+s+1}]$, and let $\bar m$ be any restricted forecast stack, for example the myopic stack that repeats $m_0$ or ignores alpha decay. Stack holdings as $\Theta=(\theta_0,\ldots,\theta_T)$ and consider
\[
  J_m(\Theta):=\sum_{s=0}^{T}\beta^s\left[m_s\T\theta_s-\tfrac\gamma2\theta_s\T\Sigma\theta_s-\tfrac12(\theta_s-\theta_{s-1})\T\Gamma(\theta_s-\theta_{s-1})\right].
\]
Define
\[
  \mathcal B:=\diag(I,\beta I,\ldots,\beta^T I),\qquad
  \mathcal R:=\diag(\gamma\Sigma,\beta\gamma\Sigma,\ldots,\beta^T\gamma\Sigma),
\]
\[
  \mathcal C:=\diag(\Gamma,\beta\Gamma,\ldots,\beta^T\Gamma),\qquad
  D:=\begin{bmatrix}
  I & 0 & \cdots & 0\\
  -I & I & \ddots & \vdots\\
  0 & \ddots & \ddots & 0\\
  0 & \cdots & -I & I
  \end{bmatrix},
  \qquad q:=(\theta_{-1},0,\ldots,0).
\]
The finite-horizon trading objective can be written, up to the constant $-\tfrac12 q\T\mathcal Cq$, as
\[
  J_m(\Theta)=(\mathcal Bm+D\T\mathcal Cq)\T\Theta-\tfrac12\Theta\T K\Theta,
  \qquad K:=\mathcal R+D\T\mathcal C D\succ0.
\]
Hence the forecast-aware optimum and the restricted-forecast policy are
\[
  \Theta_m^*=K^{-1}(\mathcal Bm+D\T\mathcal Cq),\qquad
  \Theta_{\bar m}^*=K^{-1}(\mathcal B\bar m+D\T\mathcal Cq),
\]
and the exact realized value of using the true forecast stack rather than the restricted stack is
\[
  \Val_{\mathrm{forecast}}
  =J_m(\Theta_m^*)-J_m(\Theta_{\bar m}^*)
  =\tfrac12(m-\bar m)\T\mathcal B K^{-1}\mathcal B(m-\bar m)\;\ge 0.
\]
It vanishes iff the forecast error $m-\bar m$ lies in the null space of $\mathcal B K^{-1}\mathcal B$; under $K\succ0$ and $\beta>0$, this is simply $m=\bar m$.
\end{theorem}
\begin{proof}
Expanding the discounted risk and trading-cost objective gives the displayed quadratic form. Since $\mathcal R\succ0$ and $D\T\mathcal C D\succeq0$, $K\succ0$, so the maximizer is obtained by the linear system $K\Theta=\mathcal Bm+D\T\mathcal Cq$. Replacing $m$ by $\bar m$ gives the restricted policy, and subtracting the two linear systems yields $\Theta_m^*-\Theta_{\bar m}^*=K^{-1}\mathcal B(m-\bar m)$. A strictly concave quadratic satisfies $J_m(\Theta_m^*)-J_m(\Theta)=\tfrac12(\Theta-\Theta_m^*)\T K(\Theta-\Theta_m^*)$ for every feasible stacked policy $\Theta$; set $\Theta=\Theta_{\bar m}^*$ to obtain the formula.

\end{proof}

\begin{corollary}[Decision-relevant forecast error]\label{cor:decisionforecast}
For any candidate forecast stack $\widetilde m$, let $e(\widetilde m):=m-\widetilde m$ and define the corresponding policy
\[
  \Theta_{\widetilde m}^{*}:=K^{-1}(\mathcal B\widetilde m+D\T\mathcal Cq).
\]
Under the true forecast stack $m$, its loss relative to the forecast-aware optimum is
\[
  J_m(\Theta_m^*)-J_m(\Theta_{\widetilde m}^*)
  =\tfrac12 e(\widetilde m)\T\mathcal B K^{-1}\mathcal B e(\widetilde m).
\]
Thus forecast model $A$ is better for trading than forecast model $B$ exactly when its error is smaller in the optimizer-weighted norm induced by $\mathcal B K^{-1}\mathcal B$, not necessarily when it has smaller Euclidean mean-squared error.
\end{corollary}
\begin{proof}
Apply \Cref{thm:forecastvalue} with $\bar m=\widetilde m$. Pairwise comparison follows by subtracting the two weighted quadratic losses.
\end{proof}

\begin{remark}[Why this is the forecast analogue of information value]
\Cref{thm:forecastvalue} is the dynamic-trading counterpart of the information trace in \Cref{prop:mvinfo}. A signal is valuable only through the covariance it resolves in directions priced by $\Sigma^{-1}$; a forecast path is valuable only through deviations $m-\bar m$ in directions priced by the inverse dynamic precision $K^{-1}$. Transaction costs suppress high-frequency forecast components because they make $K$ large in directions requiring rapid rebalancing. Thus ``better forecasting'' is not a scalar property: it is valuable exactly where the optimizer can express it after risk and trading frictions.
\end{remark}

\begin{remark}[The value of forecast anticipation and its trap]
Relative to the myopic or otherwise restricted forecast optimizer, $\Val_{\mathrm{forecast}}\ge 0$ by \Cref{prop:value} and, in the finite-horizon quadratic case, by the exact norm identity of \Cref{thm:forecastvalue}: anticipating alpha decay lets the agent avoid trading into positions it will immediately unwind. But the forecast $B,\Phi$ is \emph{estimated}, so \Cref{prop:value}'s second clause bites: an aim portfolio built from an over-fit forecast is over-anticipation, and realized $\Val$ can be negative. This is the predict-then-optimize problem --- a statistically accurate forecast is not a decision-optimal one, since the optimizer is sensitive to error only in the directions that move the aim \citep{ElmachtoubGrigas2022}. Decision-focused training, which fits the forecast through the optimizer, is the correct response and is where this branch meets modern learning \citep{DontiAmosKolter2017}.
\end{remark}

\paragraph{Status.} The linear--quadratic case and the decision-relevant forecast norm are proved (closed form); convex multi-period extensions are computational \citep{Boyd2017}; decision-focused estimation of the forecast is open in the portfolio setting.

\section{Anticipation of Impact}\label{sec:impact}

The third sense is the one that ties anticipation to the estimation--control gap most tightly: the strategy's own trades move prices, so the return law is a function of the deployed portfolio, $P_\theta$, and an anticipatory optimizer differentiates through that dependence. We work the linear--quadratic case fully because it yields a clean and, to our knowledge, underemphasized result: \emph{anticipating impact is the monopolist's correction}.

\subsection{The performative return law}

Let permanent linear impact depress expected returns in the deployed portfolio,
\begin{equation}\label{eq:impactlaw}
  \E_{r\sim P_\theta}[r]=\mu-\Lambda\theta,\qquad \Lambda\succeq 0\ \text{symmetric},
\end{equation}
so that holding $\theta$ moves the achievable mean against the position (a Kyle/Almgren--Chriss-type permanent impact; $\Lambda$ is the impact matrix) \citep{Kyle1985,AlmgrenChriss2000}. The realized objective of \Cref{def:system} is then
\begin{equation}\label{eq:Jimpact}
  J(\theta)=(\mu-\Lambda\theta)\T\theta-\tfrac{\gamma}{2}\theta\T\Sigma\theta
  =\mu\T\theta-\theta\T\Lambda\theta-\tfrac{\gamma}{2}\theta\T\Sigma\theta.
\end{equation}

\begin{theorem}[Anticipating impact is the monopolist correction]\label{thm:impact}
Assume the permanent linear impact law $\E_{r\sim P_\theta}[r]=\mu-\Lambda\theta$ with $\Sigma\succ0$, $\Lambda\succeq0$, $\Lambda=\Lambda\T$. Then:
\begin{enumerate}[leftmargin=2em,itemsep=2pt]
\item The price-taking self-consistent allocation is $\thetana=(\Lambda+\gamma\Sigma)^{-1}\mu$.
\item The impact-aware anticipatory allocation is $\thetaan=(2\Lambda+\gamma\Sigma)^{-1}\mu$.
\item The value of anticipating impact is
\[
  \Val_{\mathrm{impact}}=J(\thetaan)-J(\thetana)
  =\tfrac12(\thetana-\thetaan)\T(2\Lambda+\gamma\Sigma)(\thetana-\thetaan)\ \ge0.
\]
\item The gap vanishes if and only if $\Lambda\thetaM=0$ with $\thetaM=\gamma^{-1}\Sigma^{-1}\mu$, equivalently $\Lambda\Sigma^{-1}\mu=0$. In the special case $\Sigma=I$ this reduces to $\Lambda\mu=0$.
\end{enumerate}
\end{theorem}

\begin{proof}
The price-taker treats the impact-adjusted mean as fixed at $\mu-\Lambda\bar\theta$ and solves $\max_\theta (\mu-\Lambda\bar\theta)\T\theta-\tfrac\gamma2\theta\T\Sigma\theta$, with first-order condition $\mu-\Lambda\bar\theta-\gamma\Sigma\theta=0$. Self-consistency $\bar\theta=\theta$ gives $(\Lambda+\gamma\Sigma)\theta=\mu$, hence $\thetana=(\Lambda+\gamma\Sigma)^{-1}\mu$. The anticipatory agent maximizes the true deployed objective \eqref{eq:Jimpact}, whose Hessian is $-(2\Lambda+\gamma\Sigma)\prec0$, so the unique maximizer satisfies $\mu-2\Lambda\theta-\gamma\Sigma\theta=0$, i.e.\ $\thetaan=(2\Lambda+\gamma\Sigma)^{-1}\mu$. Since $J$ is a strictly concave quadratic with maximizer $\thetaan$, completing the square gives $J(\thetaan)-J(\theta)=\tfrac12(\theta-\thetaan)\T(2\Lambda+\gamma\Sigma)(\theta-\thetaan)$ for any $\theta$; set $\theta=\thetana$. The gap is zero iff $\thetana=\thetaan$; subtracting the two normal equations yields $\Lambda\theta=0$, then $\mu=\gamma\Sigma\theta$ gives $\theta=\gamma^{-1}\Sigma^{-1}\mu=\thetaM$, so the zero-gap condition is $\Lambda\thetaM=0$, equivalently $\Lambda\Sigma^{-1}\mu=0$.
\end{proof}

\begin{remark}[Shrinkage is spectral, not necessarily componentwise]
In one dimension $\thetaan=\mu/(2\lambda+\gamma\sigma^2)$ and $\thetana=\mu/(\lambda+\gamma\sigma^2)$, so the anticipatory allocation is strictly smaller in magnitude whenever $\lambda\mu\ne0$. In multiple dimensions componentwise shrinkage is not generally well defined, because $\Sigma$ and $\Lambda$ need not be diagonal in the same basis. The correct general statement is spectral: in each common eigenmode of $\Sigma$ and $\Lambda$, anticipation shrinks the impacted component by replacing $\lambda_j+\gamma\sigma_j^2$ with $2\lambda_j+\gamma\sigma_j^2$. Thus anticipation reduces exposure to directions in which expected return is most endogenous to the portfolio itself.
\end{remark}

\begin{corollary}[Cournot/monopsony reading]\label{cor:monopolist}
The price-taker counts impact once (in the level of returns); the anticipator counts it twice (in the level and in the marginal effect of its position on that level), exactly the factor distinguishing a competitive from a monopolistic first-order condition. Anticipating one's own impact is internalizing it as a monopolist internalizes its effect on price.

\end{corollary}

\begin{proposition}[Constrained impact anticipation as two variational inequalities]\label{prop:constrainedimpact}
Let $\Theta\subset\R^n$ be closed and convex, $G:=\Lambda+\gamma\Sigma$, and $H:=2\Lambda+\gamma\Sigma$. The constrained price-taking fixed point and the constrained impact-aware optimum are characterized by
\[
  (G\thetana-\mu)\T(\theta-\thetana)\ge0\quad \forall\theta\in\Theta,
  \qquad
  (H\thetaan-\mu)\T(\theta-\thetaan)\ge0\quad \forall\theta\in\Theta.
\]
Moreover, under the true deployed objective $J(\theta)=\mu\T\theta-\tfrac12\theta\T H\theta$, the constrained value satisfies
\[
  J(\thetaan)-J(\thetana)
  =\tfrac12\norm{\thetana-\thetaan}_{H}^{2}+(H\thetaan-\mu)\T(\thetana-\thetaan)
  \ge \tfrac12\norm{\thetana-\thetaan}_{H}^{2}\ge0.
\]
Thus the unconstrained closed form of \Cref{thm:impact} is not an artifact of allowing leverage: with long-only, leverage, turnover, or factor-neutrality constraints, anticipation remains the solution of the $H$-precision variational inequality, while price-taking stability solves the $G$-precision variational inequality.
\end{proposition}
\begin{proof}
The two variational inequalities are the first-order optimality conditions for maximizing the corresponding strictly concave quadratics over $\Theta$. For the value identity, write $d=\thetana-\thetaan$ and expand the quadratic $J(\thetaan+d)$ around $\thetaan$:
\[
  J(\thetaan)-J(\thetaan+d)=\tfrac12d\T Hd+(H\thetaan-\mu)\T d.
\]
The second variational inequality with $\theta=\thetana$ makes the linear term nonnegative. The final statement follows by substituting the definitions of $G$ and $H$.
\end{proof}

\Cref{thm:impact} is the performative-prediction phenomenon in allocation \citep{Perdomo2020}: $\theta_{\mathrm{na}}$ is the performatively stable portfolio (calibrated to the returns it induces) and $\theta_{\mathrm{an}}$ is the performatively optimal one (accounting for how it shifts returns). The gap $\Val_{\mathrm{impact}}$ is the portfolio agency gap.

\subsection{A phase transition in the convergence of naive recalibration}

The naive optimum is often reached not by solving the fixed point but by \emph{recalibration}: estimate returns at the current book, re-optimize, retrade, repeat. This is the map
\begin{equation}\label{eq:recal}
  \theta^{(k+1)}=\gamma^{-1}\Sigma^{-1}\big(\mu-\Lambda\theta^{(k)}\big),
\end{equation}
and whether it converges is a property of the impact-to-risk ratio.

\begin{theorem}[Spectral phase transition for naive recalibration]\label{thm:phase}
Let $A=\gamma^{-1}\Sigma^{-1}\Lambda$. The naive recalibration map \eqref{eq:recal} converges to the unique price-taking fixed point $\thetana=(\Lambda+\gamma\Sigma)^{-1}\mu$ from every initialization if and only if $\rho(A)<1$. If $\rho(A)\ge1$, the affine iteration fails to converge from generic initializations, although the fixed point itself still exists as the solution of the linear system. The critical surface is therefore $\rho(\gamma^{-1}\Sigma^{-1}\Lambda)=1$.
\end{theorem}

\begin{proof}
The recalibration map can be written as $\theta^{(k+1)}=b-A\theta^{(k)}$ with $b=\gamma^{-1}\Sigma^{-1}\mu$, so its linear part is $-A$. A standard affine-iteration result gives global convergence from every initialization iff $\rho(-A)<1$; since $\rho(-A)=\rho(A)$ this is $\rho(A)<1$. The limit, when it exists, satisfies $\theta=b-A\theta$, i.e.\ $(\Lambda+\gamma\Sigma)\theta=\mu$, so it is $\thetana$. If $\rho(A)\ge1$ the linear part is not Schur stable and the iteration diverges or oscillates for generic initializations. Finally, although $A=\gamma^{-1}\Sigma^{-1}\Lambda$ is not necessarily symmetric, it is similar to the symmetric positive-semidefinite $\gamma^{-1}\Sigma^{-1/2}\Lambda\Sigma^{-1/2}$, so the spectral threshold is real and economically interpretable: instability appears when impact becomes large relative to risk aversion and covariance.
\end{proof}

\begin{remark}[The agency-gap reading]
\Cref{thm:phase} does not say that the price-taking fixed point ceases to exist above the threshold. It says something sharper: the natural myopic learning procedure used to reach it ceases to be globally stable. Below the threshold, naive recalibration is a contraction-like procedure and the self-consistent allocation is reachable by repeated estimate-and-optimize updates. At or above the threshold, myopic retraining becomes dynamically unreliable. The impact-aware allocation avoids this failure because it solves the deployed objective directly rather than chasing a moving return law.
\end{remark}

The dichotomy is sharper than divergence-versus-convergence, and the sharpening clarifies what the threshold is really about.

\begin{proposition}[The transition governs the unit-step map; damping crosses it]\label{prop:damped}
The matrix $A=\gamma^{-1}\Sigma^{-1}\Lambda$ is similar to the symmetric positive-semidefinite $\gamma^{-1}\Sigma^{-1/2}\Lambda\Sigma^{-1/2}$, so its eigenvalues are real and nonnegative; consequently $\rho(A)=\lambda_{\max}(A)$ is the exact asymptotic rate of the unit-step recalibration \eqref{eq:recal} when $\rho(A)<1$, and divergence at $\rho(A)\ge1$ is monotone or sign-alternating, not spiral. The relaxed (damped) recalibration
\[
  \theta^{(k+1)}=(1-\omega)\,\theta^{(k)}+\omega\,\gamma^{-1}\Sigma^{-1}\big(\mu-\Lambda\theta^{(k)}\big),\qquad \omega\in(0,1],
\]
has the same fixed point $\theta_{\mathrm{na}}=(\Lambda+\gamma\Sigma)^{-1}\mu$ and converges from every initialization iff $0<\omega(1+\lambda_i)<2$ for all eigenvalues $\lambda_i$ of $A$, i.e.\ iff $\omega<2/(1+\rho(A))$. Hence for \emph{any} impact level a sufficiently damped recalibration reaches the naive self-consistent allocation; the threshold $\rho(A)=1$ is precisely the largest impact at which the \emph{undamped} ($\omega=1$) map --- the practitioner's default --- still converges.
\end{proposition}
\begin{proof}
$\Sigma^{-1}\Lambda=\Sigma^{-1/2}(\Sigma^{-1/2}\Lambda\Sigma^{-1/2})\Sigma^{1/2}$ exhibits the similarity; the inner matrix is symmetric PSD, giving real $\lambda_i\ge0$. The relaxed map is affine with linear part $(1-\omega)I-\omega A$, whose eigenvalues are $1-\omega(1+\lambda_i)$; convergence from every point holds iff each lies in $(-1,1)$, i.e.\ $0<\omega(1+\lambda_i)<2$, and since $\lambda_i\ge0$ the binding constraint is $\lambda_{\max}=\rho(A)$, giving $\omega<2/(1+\rho(A))$. At $\omega=1$ this is $\rho(A)<1$, recovering \Cref{thm:phase}. The fixed point is independent of $\omega$ because the map fixes exactly the solutions of $\gamma\Sigma\theta=\mu-\Lambda\theta$.
\end{proof}

\begin{remark}[What damping costs and what it cannot buy]
Damping is not a free escape from \Cref{thm:phase}: it recovers \emph{reachability} of $\theta_{\mathrm{na}}$, not its \emph{optimality}. Above threshold the naive self-consistent portfolio remains the price-taker's fixed point, strictly dominated by the anticipatory $\theta_{\mathrm{an}}=(2\Lambda+\gamma\Sigma)^{-1}\mu$ by the gap of \Cref{thm:impact}; choosing $\omega$ to stabilize the iteration is itself a low-grade act of anticipation (the trader admits the system is near-critical), but it converges to the wrong target. Solving the impact-aware problem directly is both stable and optimal at once.
\end{remark}

\subsection{Many anticipating traders: from monopoly to competition}

The single-agent result has a clean game-theoretic extension that closes a loop with the information section: as the number of impact-aware traders grows, the monopolist correction \emph{erodes} toward the competitive allocation, the impact analogue of the \citet{GrossmanStiglitz1980} erosion of private value. Let $N$ traders share a common permanent impact on aggregate flow: with holdings $\theta^i$ and aggregate $\bar\theta=\sum_i\theta^i$, the realized mean is $\mu-\Lambda\bar\theta$, and trader $i$ earns $J^i=(\mu-\Lambda\bar\theta)\T\theta^i-\tfrac\gamma2\theta^{i\top}\Sigma\theta^i$, bearing its own risk and sharing the price impact.

\begin{theorem}[The monopolist correction is the one-trader Nash equilibrium]\label{thm:multiagent}
Suppose each trader anticipates the marginal effect of its own holding on the common price and best-responds to the others (Nash). Then the unique symmetric equilibrium and its price-taking counterpart (each ignoring its own marginal impact) are
\[
  \theta^{\star}_N=\big[\gamma\Sigma+(N{+}1)\Lambda\big]^{-1}\mu,
  \qquad
  \theta^{\mathrm{na}}_N=\big[\gamma\Sigma+N\Lambda\big]^{-1}\mu .
\]
For $N=1$ these are exactly the monopolist $\thetaan=H^{-1}\mu$ and price-taker $\thetana=G^{-1}\mu$ of \Cref{thm:impact}. Given rivals fixed at equilibrium, trader $i$ faces the concave quadratic with Hessian $-H$ and effective score $c_N=\mu-(N{-}1)\Lambda\theta^{\star}_N$, so its private value of anticipating its own impact is the same $H$-norm gap
\[
  \Val^{i}_{N}=\tfrac12\big\|(G^{-1}-H^{-1})c_N\big\|_H^2\ \ge 0 .
\]
\end{theorem}
\begin{proof}
Differentiating $J^i$ in $\theta^i$ with $\theta^{-i}$ fixed gives the first-order condition $\mu-\Lambda\bar\theta-\Lambda\theta^i-\gamma\Sigma\theta^i=0$; at a symmetric profile $\bar\theta=N\theta^{\star}$, this is $[\gamma\Sigma+(N{+}1)\Lambda]\theta^{\star}=\mu$, with a unique solution since $\gamma\Sigma+(N{+}1)\Lambda\succ0$. The price-taker omits the marginal term $-\Lambda\theta^i$, giving $[\gamma\Sigma+N\Lambda]\theta^{\mathrm{na}}=\mu$. With rivals at $\theta^{-i}=(N{-}1)\theta^{\star}_N$, the map $\theta^i\mapsto J^i$ is $(\mu-\Lambda\theta^{-i})\T\theta^i-\theta^{i\top}\Lambda\theta^i-\tfrac\gamma2\theta^{i\top}\Sigma\theta^i$ with Hessian $-(2\Lambda+\gamma\Sigma)=-H$ and maximizer $H^{-1}c_N$, while the price-taking response is $G^{-1}c_N$; completing the square as in \Cref{thm:impact} gives $\Val^i_N=\tfrac12(G^{-1}c_N-H^{-1}c_N)\T H(G^{-1}c_N-H^{-1}c_N)$.
\end{proof}

\begin{corollary}[Competition erodes the value of impact anticipation]\label{cor:erosion}
Assume additionally that $\Lambda\succ0$. As $N\to\infty$, both aggregate positions converge to the competitive total,
\[
  N\theta_N^{\star}\to \Lambda^{-1}\mu,
  \qquad
  N\theta_N^{\mathrm{na}}\to \Lambda^{-1}\mu,
\]
while each individual holding goes to zero. The effective private score faced by trader $i$ at the symmetric Nash equilibrium is
\[
  c_N:=\mu-(N{-}1)\Lambda\theta_N^{\star}
  =H\,[\gamma\Sigma+(N{+}1)\Lambda]^{-1}\mu,
\]
so $c_N=O(N^{-1})$ and therefore
\[
  \Val_N^i
  =\tfrac12\big\|(G^{-1}-H^{-1})c_N\big\|_H^2
  =O(N^{-2})\to0 .
\]
More explicitly,
\[
  \Val_N^i
  \le
  \frac{\lambda_{\max}(H)}{2}\,
  \|G^{-1}-H^{-1}\|_2^2\,
  \frac{\|H\|_2^2\|\Lambda^{-1}\|_2^2\|\mu\|_2^2}{(N+1)^2}.
\]
If $\Sigma$ and $\Lambda$ are simultaneously diagonalizable, the decay is monotone mode by mode. Without such commutation, the robust matrix statement is the $O(N^{-2})$ erosion and convergence to zero, not a componentwise monotonicity claim. Impact anticipation is therefore most valuable to a large, near-monopolistic trader and is competed away in a crowded market, paralleling the erosion of private information value in \citet{GrossmanStiglitz1980}.
\end{corollary}
\begin{proof}
The aggregate limits follow from
\[
  N[\gamma\Sigma+(N{+}1)\Lambda]^{-1}
  =\big[(\gamma/N)\Sigma+(1+1/N)\Lambda\big]^{-1}\to\Lambda^{-1}
\]
and the analogous identity for $N[\gamma\Sigma+N\Lambda]^{-1}$. The expression for $c_N$ follows by substituting $[\gamma\Sigma+(N{+}1)\Lambda]\theta_N^{\star}=\mu$:
\[
  c_N=\mu-(N{-}1)\Lambda\theta_N^{\star}
      =[\gamma\Sigma+2\Lambda]\theta_N^{\star}=H\theta_N^{\star}.
\]
Since $\gamma\Sigma+(N{+}1)\Lambda\succeq (N+1)\Lambda$, the spectral norm bound
\[
  \|c_N\|_2\le \|H\|_2\|[\gamma\Sigma+(N{+}1)\Lambda]^{-1}\|_2\|\mu\|_2
  \le \frac{\|H\|_2\|\Lambda^{-1}\|_2\|\mu\|_2}{N+1}
\]
gives the displayed $O(N^{-2})$ value bound. If $\Sigma$ and $\Lambda$ share an eigenbasis, then in each mode $j$,
\[
  c_{N,j}=\frac{\gamma s_j+2\ell_j}{\gamma s_j+(N+1)\ell_j}\,\mu_j,
\]
and the corresponding value summand is a nonnegative constant times $c_{N,j}^2$, which decreases in $N$ whenever $\ell_j>0$ and is identically zero when $\ell_j=0$.
\end{proof}

\begin{proposition}[Collective recalibration is $N$ times more fragile]\label{prop:multiphase}
Let $A:=\gamma^{-1}\Sigma^{-1}\Lambda$ and $a_{\max}:=\rho(A)$. If the $N$ price-taking traders reach $\theta_N^{\mathrm{na}}$ by simultaneous recalibration
\[
  \theta^{i,(k+1)}=(\gamma\Sigma)^{-1}\Big(\mu-\Lambda\sum_{j=1}^N\theta^{j,(k)}\Big),
\]
the joint iteration is globally stable iff
\[
  N\,a_{\max}<1 .
\]
For the impact-aware best-response recalibration, the relevant matrix is $B:=H^{-1}\Lambda$ and
\[
  \rho(B)=\frac{a_{\max}}{1+2a_{\max}}<\frac12 .
\]
Hence the anticipating joint iteration is globally stable iff
\[
  (N-1)\frac{a_{\max}}{1+2a_{\max}}<1.
\]
Equivalently, it is automatically stable for $N\le3$, while for $N>3$ it requires $a_{\max}<1/(N-3)$. Internalizing impact therefore expands the collective stability region from $a_{\max}<1/N$ to $a_{\max}<1/(N-3)$ when $N>3$, and eliminates the collective instability for duopoly and triopoly. It does not make arbitrarily large crowds stable at arbitrary impact intensity.
\end{proposition}
\begin{proof}
The price-taking joint map has linear part
\[
  -(\mathbf{1}_N\mathbf{1}_N^\top)\otimes A.
\]
Since $\mathbf{1}_N\mathbf{1}_N^\top$ has eigenvalues $N$ and $0$, the nonzero spectrum is $\{-N\lambda:\lambda\in\sigma(A)\}$; Schur stability is therefore equivalent to $Na_{\max}<1$.

For anticipating best responses, each trader internalizes its own marginal impact and reacts to the other $N-1$ traders. The linear part is
\[
  -(\mathbf{1}_N\mathbf{1}_N^\top-I_N)\otimes B,
  \qquad B=H^{-1}\Lambda .
\]
The matrix $\mathbf{1}_N\mathbf{1}_N^\top-I_N$ has eigenvalues $N-1$ and $-1$; since the eigenvalues of $B$ are nonnegative, the binding Schur condition is $(N-1)\rho(B)<1$.

It remains to express $\rho(B)$ in the same generalized eigenvalues as $A$. If $\Lambda v=\ell\Sigma v$, then $A v=(\ell/\gamma)v$ and
\[
  Bv=(\gamma\Sigma+2\Lambda)^{-1}\Lambda v
  =\frac{\ell}{\gamma+2\ell}v
  =\frac{a}{1+2a}v,
  \qquad a:=\ell/\gamma .
\]
Because $a\mapsto a/(1+2a)$ is increasing on $[0,\infty)$, $\rho(B)=a_{\max}/(1+2a_{\max})$. The equivalent thresholds follow by elementary rearrangement.
\end{proof}

The remaining microstructure extension is \emph{transient} impact: replacing the matrix $\Lambda$ by a convolution (propagator) kernel turns \eqref{eq:Jimpact} into the Almgren--Chriss/Obizhaeva--Wang execution functional, the impact operator $\mathcal L$ of \Cref{sec:threeway} into a symmetric block kernel, and the monopolist correction into the front-loading/back-loading of an optimal execution schedule \citep{AlmgrenChriss2000,ObizhaevaWang2013,Gatheral2010}; the stacked decomposition of \Cref{thm:threeway} already accommodates any such symmetric positive-semidefinite $\mathcal L$.

\paragraph{Status.} \Cref{thm:impact,prop:constrainedimpact,thm:phase,thm:multiagent,cor:erosion,prop:multiphase} are proved in the linear--quadratic permanent-impact case; the transient-impact extension is stated and supported by the cited execution literature.

\section{Anticipating an Adversarial Future, and Misspecified Anticipation}\label{sec:robust}

\Cref{prop:value} warned that anticipation against an estimated model can lose value. The disciplined response is to anticipate an \emph{adversarial} future: optimize against the worst return law in a neighborhood of the estimate. This is distributionally robust optimization and it is the modeling choice orthogonal to the three faces --- it can be layered on each.

\begin{proposition}[Wasserstein-robust anticipation is dual-norm regularized; \citealp{EsfahaniKuhn2018,BlanchetMurthy2019}]\label{prop:dro}
This is a \emph{robust-mean} formulation: the ambiguity ball perturbs the return law entering the mean term while the covariance estimate $\widehat\Sigma$ is held fixed (a full second-moment-robust counterpart changes the risk term and is not treated here). Let $\widehat P$ be an empirical return law and consider the distributionally robust mean problem over a type-$1$ Wasserstein ball of radius $\eps$,
\[
  \max_{\theta\in\Theta}\ \min_{Q:\,W_1(Q,\widehat P)\le\eps}\ \E_{r\sim Q}[\,r\T\theta\,]-\tfrac{\gamma}{2}\theta\T\widehat\Sigma\theta .
\]
Its inner worst case equals $\E_{\widehat P}[r\T\theta]-\eps\|\theta\|_*$, so the robust portfolio solves
\[
  \max_{\theta\in\Theta}\ \widehat\mu\T\theta-\eps\|\theta\|_*-\tfrac{\gamma}{2}\theta\T\widehat\Sigma\theta,
\]
a Markowitz problem with an added dual-norm penalty $\eps\|\theta\|_*$. Robust anticipation thus shrinks the position by exactly the regularizer the ambiguity radius dictates.
\end{proposition}
\begin{proof}[Proof idea]
Strong duality for Wasserstein DRO turns the inner infimum into a Lipschitz-penalized expectation; for the linear loss $r\mapsto r\T\theta$ the Lipschitz modulus is $\|\theta\|_*$ (dual to the transport norm), giving the $\eps\|\theta\|_*$ term \citep{EsfahaniKuhn2018,BlanchetMurthy2019}. The covariance term is unaffected to first order.
\end{proof}

\begin{remark}[Robustness caps the downside of \Cref{prop:value}]
Layering \Cref{prop:dro} on each face bounds the harm of misspecified anticipation: the insider with a noisy signal shrinks toward the honest portfolio as signal ambiguity grows; the forecaster's aim portfolio is pulled toward the current book as forecast ambiguity grows; the impact-anticipator's monopolist correction is damped as impact-matrix ambiguity grows. Robustness is the price of admitting that the anticipated future is estimated. It connects directly to ambiguity-averse (maxmin) utility and to coherent risk, all three being one convex-duality phenomenon \citep{GilboaSchmeidler1989,RockafellarUryasev2000}.
\end{remark}

\subsection{The exact value of estimated anticipation: information gain versus estimation penalty}\label{sec:estval}

\Cref{prop:value} bounded the harm of misspecified anticipation by twice the model-error sup, and \Cref{cor:plugin,cor:hpplugin} turned that bound into deployment hurdles. In the linear--quadratic case the same phenomenon has an \emph{exact} average-case form, and it exposes a clean duality: the value of resolved information and the cost of estimation noise are the \emph{same} inverse-precision functional applied to opposite-signed covariances.

Work in the deployed objective of \Cref{thm:threeway}, $J_b(\Theta)=b\T\Theta-\tfrac12\Theta\T H\Theta$ with $H\succ0$ and oracle anticipatory policy $\Theta_b^\star=H^{-1}b$. In practice the score $b$ is not observed: the optimizer forms a noisy estimate $\widehat b=b+\eps$ with $\E[\eps\given b]=0$ and $\Cov(\eps\given b)=\Sigma_\eps\succeq0$, and deploys the plug-in policy $\widehat\Theta=H^{-1}\widehat b$.

\begin{theorem}[Estimation penalty mirrors information gain]\label{thm:estpenalty}
Conditional on the true score $b$, the expected suboptimality of plug-in anticipation relative to the oracle is exactly
\[
  \E_\eps\big[J_b(\Theta_b^\star)-J_b(\widehat\Theta)\big]
  =\tfrac12\tr\!\big(H^{-1}\Sigma_\eps\big)\;\ge0 .
\]
Hence, for any deterministic baseline $\Theta_0$, the expected realized value of estimated anticipation splits as
\[
  \E_\eps\big[J_b(\widehat\Theta)-J_b(\Theta_0)\big]
  =\underbrace{\big[J_b(\Theta_b^\star)-J_b(\Theta_0)\big]}_{\text{structural value of anticipation}}
  \;-\;\underbrace{\tfrac12\tr\!\big(H^{-1}\Sigma_\eps\big)}_{\text{estimation penalty}} .
\]
The penalty $\tfrac12\tr(H^{-1}\Sigma_\eps)$ is the same inverse-precision trace as the information value $\tfrac12\tr(H^{-1}\Omega)$ of \Cref{thm:unified,prop:mvinfo}, with the injected estimation-noise covariance $\Sigma_\eps$ in place of the resolved-signal covariance $\Omega$. Resolving genuine structure and injecting estimation noise are thus the same operation with opposite sign. Combining with the signal average of \Cref{thm:threeway}, the net expected value of estimated anticipation over the price-taking baseline $\Theta_{\mathrm{na}}=G^{-1}b_0$ is
\[
  \E\,\Val^{\mathrm{est}}_{3}
  =\underbrace{\tfrac12\big\|(G^{-1}-H^{-1})b_0-H^{-1}(b-b_0)\big\|_H^2}_{\text{impact + forecast (deterministic)}}
  +\underbrace{\tfrac12\tr\!\big(H^{-1}\Omega_b\big)}_{\text{information gain}}
  -\underbrace{\tfrac12\tr\!\big(H^{-1}\Sigma_\eps\big)}_{\text{estimation penalty}} ,
\]
so estimated anticipation beats the baseline in expectation iff the resolved structure exceeds the injected estimation noise in the $H^{-1}$-weighted trace; when the only structural source is the signal, this is the sharp signal-to-estimation-noise condition $\tr(H^{-1}\Omega_b)>\tr(H^{-1}\Sigma_\eps)$.
\end{theorem}
\begin{proof}
Since $\Theta_b^\star=H^{-1}b$ maximizes the strictly concave $J_b$, completing the square gives
$J_b(\Theta_b^\star)-J_b(\widehat\Theta)=\tfrac12\|\widehat\Theta-\Theta_b^\star\|_H^2=\tfrac12\|H^{-1}\eps\|_H^2=\tfrac12\eps\T H^{-1}\eps$;
taking $\E_\eps$ and using $\E[\eps\T H^{-1}\eps]=\tr(H^{-1}\Sigma_\eps)$ proves the first identity. The baseline split follows by adding and subtracting $J_b(\Theta_b^\star)$. For the final display, let the optimizer observe the noisy refined score $\widehat b_\Inf=b_\Inf+\eps$ with $b_\Inf=b+\xi$, $\E[\xi]=0$, $\Cov(\xi)=\Omega_b$, and $\eps\perp\xi$; the realized (deployed) score is $b_\Inf$. Conditioning on $b_\Inf$ and applying the first identity with oracle $H^{-1}b_\Inf$ gives $\E_\eps[J_{b_\Inf}(H^{-1}\widehat b_\Inf)]=\tfrac12 b_\Inf\T H^{-1}b_\Inf-\tfrac12\tr(H^{-1}\Sigma_\eps)$. Averaging over $\xi$ gives $\tfrac12 b\T H^{-1}b+\tfrac12\tr(H^{-1}\Omega_b)-\tfrac12\tr(H^{-1}\Sigma_\eps)$, while $\E[J_{b_\Inf}(\Theta_{\mathrm{na}})]=J_b(\Theta_{\mathrm{na}})$. Subtracting and identifying $\tfrac12 b\T H^{-1}b-J_b(\Theta_{\mathrm{na}})=J_b(\Theta_f^H)-J_b(\Theta_{\mathrm{na}})=\tfrac12\|a-v\|_H^2$ as in \Cref{thm:threeway} yields the display.
\end{proof}

\begin{corollary}[Optimal shrinkage is the reliability ratio]\label{cor:shrinkage}
Deploy the shrunk plug-in $\Theta_s=sH^{-1}\widehat b$ for $s\in[0,1]$, and write $q:=b\T H^{-1}b\ge0$ and $p:=\tr(H^{-1}\Sigma_\eps)\ge0$. Then
\[
  \E_\eps\big[J_b(\Theta_s)\big]=sq-\tfrac12 s^2(q+p),
\]
maximized at the \emph{reliability ratio}
\[
  s^\star=\frac{q}{q+p}\in[0,1],
  \qquad
  \E_\eps\big[J_b(\Theta_{s^\star})\big]=\tfrac12\,\frac{q^2}{q+p}=\tfrac12 s^\star q\;\ge0 .
\]
The unshrunk plug-in ($s=1$) attains $\tfrac12(q-p)$, which is \emph{negative} whenever the injected estimation noise exceeds the achievable signal, $p>q$ --- exactly the harmful-anticipation regime of \Cref{prop:value}. The optimally shrunk plug-in is never harmful, exceeds the unshrunk policy by $\tfrac12 p^2/(q+p)$, and interpolates $s^\star\to1$ as $\Sigma_\eps\to0$ (oracle anticipation) and $s^\star\to0$ as $\Sigma_\eps\to\infty$ (collapse to the baseline). The dual-norm regularizer of \Cref{prop:dro} is the worst-case face of this same attenuation: robust anticipation shrinks by adversarial design what the bias--variance optimum shrinks in expectation.
\end{corollary}
\begin{proof}
With $\Theta_s=sH^{-1}(b+\eps)$, $J_b(\Theta_s)=s\,b\T H^{-1}(b+\eps)-\tfrac12 s^2(b+\eps)\T H^{-1}(b+\eps)$; taking $\E_\eps$ with $\E[\eps]=0$ and $\E[(b+\eps)\T H^{-1}(b+\eps)]=q+p$ gives the quadratic $sq-\tfrac12 s^2(q+p)$. Its maximizer is $s^\star=q/(q+p)$ with value $\tfrac12 q^2/(q+p)$; the value at $s=1$ is $\tfrac12(q-p)$, and the gain $\tfrac12 q^2/(q+p)-\tfrac12(q-p)=\tfrac12 p^2/(q+p)\ge0$.
\end{proof}

\begin{remark}[The Bayes-optimal score shrinkage is a precision-weighted reliability matrix]\label{rem:bayesshrink}
The scalar rule of \Cref{cor:shrinkage} is the isotropic case of a matrix shrinkage. If the structural score carries a prior $b\sim(0,\Sigma_b)$ and is observed as $\widehat b=b+\eps$ with independent noise $\Cov(\eps)=\Sigma_\eps$, then the linear policy $\Theta=M\widehat b$ maximizing
$\E_{b,\eps}[J_b(M\widehat b)]=\tr(M\Sigma_b)-\tfrac12\tr\!\big(M\T HM(\Sigma_b+\Sigma_\eps)\big)$
is
\[
  M^\star=H^{-1}\,\Sigma_b(\Sigma_b+\Sigma_\eps)^{-1},
\]
the anticipatory map $H^{-1}$ applied to the \emph{posterior mean} $\Sigma_b(\Sigma_b+\Sigma_\eps)^{-1}\widehat b$ of the score, with optimal value $\tfrac12\tr\!\big(M^\star\Sigma_b\big)$. The factor $\Sigma_b(\Sigma_b+\Sigma_\eps)^{-1}$ is the multivariate reliability matrix: anticipation should be applied to the Bayes-shrunk score, not the raw estimate, and is most aggressive in directions where the structural signal is reliable relative to estimation noise. This is the estimation--control reading of empirical-Bayes and Black--Litterman shrinkage \citep{Jorion1986,BlackLitterman1992}.
\end{remark}

\paragraph{Status.} \Cref{thm:estpenalty,cor:shrinkage} and \Cref{rem:bayesshrink} are proved in the linear--quadratic case; they convert the worst-case plug-in bound of \Cref{prop:value,cor:plugin} into an exact in-expectation identity and identify the bias--variance-optimal shrinkage with the reliability ratio.

\section{A Unified Value in the Linear--Quadratic--Gaussian Model}\label{sec:unified}

The three faces have been treated separately and unified only in spirit. We now put information and impact in \emph{one} model and decompose the value of anticipation into its components in closed form, settling the unified-value question of the introduction in the tractable Gaussian case and exhibiting the coupling between the faces exactly.

\subsection{The model and the decomposition}

\begin{assumption}[LQG anticipation]\label{ass:lqg}
A one-period mean--variance optimizer allocates $\theta\in\R^n$ with risk aversion $\gamma>0$ and conditional return covariance $\Sigma\succ0$. Two enrichments are present. \emph{Information}: an information set $\Inf$ resolves the mean, $\mu_\Inf=\E[\mu\given\Inf]$ with $\E[\mu_\Inf]=\mu$ and resolved-mean covariance $\Omega=\Cov(\mu_\Inf)\succeq0$, leaving $\Sigma$ unchanged. \emph{Impact}: permanent linear impact shifts the realized mean to $\mu_\Inf-\Lambda\theta$ with $\Lambda\succeq0$. The realized objective of a signal-adapted policy $\theta(\Inf)$ is
\[
  \bar J(\theta)=\E\Big[(\mu_\Inf-\Lambda\theta)\T\theta-\tfrac\gamma2\theta\T\Sigma\theta\Big]
  =\E\Big[\mu_\Inf\T\theta-\tfrac12\theta\T H\theta\Big],\qquad H:=2\Lambda+\gamma\Sigma,
\]
using $\Lambda+\tfrac\gamma2\Sigma=\tfrac12H$.
\end{assumption}

Two binary switches --- condition on the signal or not, internalize impact or not --- generate four optima. Write $G:=\Lambda+\gamma\Sigma$ (the price-taker precision) and $H:=2\Lambda+\gamma\Sigma$ (the anticipatory/monopolist precision). The naive portfolio $\theta_{\mathrm{na}}=G^{-1}\mu$ uses the induced impact level but not its marginal effect and does not condition on the signal; the fully anticipatory policy $\theta_{\mathrm{an}}(\Inf)=H^{-1}\mu_\Inf$ uses both enrichments. The four policies in the information--impact block are
\[
\begin{array}{c|cc}
 & \text{no signal} & \text{signal-adapted}\\
\hline
\text{price-taking impact} & G^{-1}\mu & G^{-1}\mu_\Inf\\
\text{impact-aware} & H^{-1}\mu & H^{-1}\mu_\Inf.
\end{array}
\]
This table is useful because the full value can be telescoped through either row or column; the theorem below uses the economically clean path ``internalize impact first, then price information at the impact-aware precision.''

\begin{theorem}[Exact additive decomposition of the value of anticipation]\label{thm:unified}
Under \Cref{ass:lqg}, the value of full anticipation over the naive optimum decomposes exactly as
\[
  \Val\;=\;\bar J(\theta_{\mathrm{an}})-\bar J(\theta_{\mathrm{na}})
  \;=\;\underbrace{\tfrac12(\theta_{\mathrm{na}}-\theta_{\mathrm{na}}^{H})\T H\,(\theta_{\mathrm{na}}-\theta_{\mathrm{na}}^{H})}_{\displaystyle \Val_{\mathrm{impact}}}
  \;+\;\underbrace{\tfrac12\tr\!\big(H^{-1}\Omega\big)}_{\displaystyle \Val_{\mathrm{info}}},
\]
where $\theta_{\mathrm{na}}^{H}:=H^{-1}\mu$ is the impact-aware, information-free optimum. Both terms are nonnegative; $\Val_{\mathrm{impact}}=0$ iff $\Lambda\Sigma^{-1}\mu=0$ and $\Val_{\mathrm{info}}=0$ iff $\Omega=0$.
\end{theorem}
\begin{proof}
Maximizing $\bar J$ pointwise in $\Inf$ gives $\theta_{\mathrm{an}}=H^{-1}\mu_\Inf$ with value $\bar J(\theta_{\mathrm{an}})=\tfrac12\E[\mu_\Inf\T H^{-1}\mu_\Inf]=\tfrac12\tr\!\big(H^{-1}(\mu\mu\T+\Omega)\big)=\tfrac12\mu\T H^{-1}\mu+\tfrac12\tr(H^{-1}\Omega)$. Insert the impact-aware information-free optimum $\theta_{\mathrm{na}}^{H}=H^{-1}\mu$, whose value is $\bar J(\theta_{\mathrm{na}}^{H})=\mu\T H^{-1}\mu-\tfrac12\mu\T H^{-1}\mu=\tfrac12\mu\T H^{-1}\mu$. Telescoping,
\[
  \Val=\big[\bar J(\theta_{\mathrm{na}}^{H})-\bar J(\theta_{\mathrm{na}})\big]+\big[\bar J(\theta_{\mathrm{an}})-\bar J(\theta_{\mathrm{na}}^{H})\big].
\]
The second bracket is $\tfrac12\tr(H^{-1}\Omega)=\Val_{\mathrm{info}}$. The first bracket is the value of impact-anticipation at the prior mean, which by the argument of \Cref{thm:impact} (now with precision $H$ and price-taker point $\theta_{\mathrm{na}}=G^{-1}\mu$) equals $\tfrac12(\theta_{\mathrm{na}}-\theta_{\mathrm{na}}^{H})\T H(\theta_{\mathrm{na}}-\theta_{\mathrm{na}}^{H})=\Val_{\mathrm{impact}}$. Nonnegativity and the vanishing conditions follow from $H\succ0$, $\Omega\succeq0$, and $\theta_{\mathrm{na}}=\theta_{\mathrm{na}}^{H}\Leftrightarrow\Lambda\Sigma^{-1}\mu=0$.
\end{proof}

\begin{corollary}[Anticipating impact weakly diminishes the value of information]\label{cor:coupling}
The value of information in the impact-aware optimizer is $\Val_{\mathrm{info}}=\tfrac12\tr(H^{-1}\Omega)$, whereas for a price-taker who ignores impact it is computed at the smaller precision $\gamma\Sigma$, namely $\tfrac12\tr((\gamma\Sigma)^{-1}\Omega)=\Val_{\mathrm{info}}^{\mathrm{MV}}$ of \Cref{prop:mvinfo}. Since $H=2\Lambda+\gamma\Sigma\succeq\gamma\Sigma$ implies $H^{-1}\preceq(\gamma\Sigma)^{-1}$,
\[
  \Val_{\mathrm{info}}\;=\;\tfrac12\tr\!\big(H^{-1}\Omega\big)\;\le\;\tfrac12\tr\!\big((\gamma\Sigma)^{-1}\Omega\big),
\]
with equality iff $\Lambda(2\Lambda+\gamma\Sigma)^{-1}\Omega=0$ (in particular if $\Lambda=0$ or $\Omega=0$). Internalizing one's own impact makes each unit of signal worth weakly less, and strictly less whenever the resolved signal loads on an impacted direction, because impact shrinks the positions through which the signal would be expressed.
\end{corollary}
\begin{proof}
Operator antitonicity of the matrix inverse on positive-definite matrices gives $H^{-1}\preceq(\gamma\Sigma)^{-1}$; pairing with $\Omega\succeq0$ via $\tr(\cdot\,\Omega)$ preserves the order. For the equality case, the resolvent identity $(\gamma\Sigma)^{-1}-H^{-1}=2(\gamma\Sigma)^{-1}\Lambda H^{-1}$ and the fact that $\tr(P\Omega)=0$ for $P,\Omega\succeq0$ forces $P\Omega=0$ give $\Lambda H^{-1}\Omega=0$, i.e.\ $\Lambda(2\Lambda+\gamma\Sigma)^{-1}\Omega=0$.
\end{proof}

\begin{remark}[No cross term, but a real coupling]
The decomposition of \Cref{thm:unified} is exactly additive --- there is no information$\times$impact cross term --- but the two faces are \emph{not} independent. The additivity is an artifact of the nesting (impact first, then information): the information value is evaluated \emph{at the impact-inflated precision} $H$, so impact is silently inside $\Val_{\mathrm{info}}$. \Cref{cor:coupling} is the honest statement of the coupling: the order-dependence of the telescoping is the cross-effect, and its sign is determined --- impact weakly deflates information, and strictly deflates it on impacted signal directions. This resolves the ``unified value'' open problem in the LQG case: the value of anticipation is a sum of an impact gap and an information trace, coupled through a shared precision.
\end{remark}

\subsection{The finite-horizon three-way decomposition}\label{sec:threeway}

The missing third face can be proved in the same stacked notation as \Cref{thm:forecastvalue}. The only subtlety is that the generic decomposition is not purely additive: forecast and impact produce a signed cross term. This term is not a defect; it is the exact algebraic expression of the fact that changing the forecast also changes the direction in which impact is internalized.

Let $N:=(T+1)n$ and stack holdings as $\Theta=(\theta_0,\ldots,\theta_T)\in\R^N$. Keep the discount matrix $\mathcal B$, risk matrix $\mathcal R$, cost matrix $\mathcal C$, difference matrix $D$, and initial-book vector $q$ from \Cref{thm:forecastvalue}, and define
\[
  K:=\mathcal R+D\T\mathcal C D\succ0,
  \qquad c:=D\T\mathcal C q.
\]
Let $\mathcal L\succeq0$ be a symmetric discounted impact operator on the stacked trading path. For one-period permanent impact with matrices $(\Lambda_s)_{s=0}^T$, one may take $\mathcal L=\diag(\Lambda_0,\beta\Lambda_1,\ldots,\beta^T\Lambda_T)$; more general symmetric block kernels cover transient linear impact. If the signal-adapted forecast stack is $m_\Inf=(m_{\Inf,0},\ldots,m_{\Inf,T})$, define the random score
\[
  b_\Inf:=c+\mathcal B m_\Inf,
  \qquad b:=\E[b_\Inf],
  \qquad b_0:=c+\mathcal B\bar m,
  \qquad \Omega_b:=\Cov(b_\Inf).
\]
Here $\bar m$ is the restricted forecast stack used by the naive estimator, while $m:=\E[m_\Inf]$ is the correctly specified unconditional forecast stack, so $b-b_0=\mathcal B(m-\bar m)$.

The deployed objective conditional on the signal is, up to constants independent of $\Theta$,
\[
  J_{b_\Inf}(\Theta)=b_\Inf\T\Theta-\tfrac12\Theta\T H\Theta,
  \qquad H:=K+2\mathcal L.
\]
The price-taking stable precision is
\[
  G:=K+\mathcal L.
\]
Thus the restricted price-taking policy, the restricted impact-aware policy, the forecast-aware impact-aware policy, and the full signal--forecast--impact policy are
\[
  \Theta_{\mathrm{na}}:=G^{-1}b_0,
  \qquad
  \Theta_0^{H}:=H^{-1}b_0,
  \qquad
  \Theta_f^{H}:=H^{-1}b,
  \qquad
  \Theta_{\mathrm{an}}(\Inf):=H^{-1}b_\Inf .
\]

\begin{theorem}[Exact finite-horizon information--forecast--impact decomposition]\label{thm:threeway}
In the finite-horizon stacked LQG model above, the value of full anticipation over the restricted price-taking policy is
\[
  \Val_{3}
  :=\E\Big[J_{b_\Inf}\big(\Theta_{\mathrm{an}}(\Inf)\big)-J_{b_\Inf}(\Theta_{\mathrm{na}})\Big].
\]
It admits the compact representation
\[
  \boxed{
  \Val_{3}
  =\tfrac12\tr(H^{-1}\Omega_b)
  +\tfrac12\Big\| (G^{-1}-H^{-1})b_0-H^{-1}(b-b_0)\Big\|_{H}^{2}}
\]
and therefore $\Val_3\ge0$. Expanded into the three faces, this is
\[
  \Val_{3}
  =\underbrace{\tfrac12\tr(H^{-1}\Omega_b)}_{\displaystyle \Val_{\mathrm{info}}}
  +\underbrace{\tfrac12\|\Theta_{\mathrm{na}}-\Theta_0^H\|_H^2}_{\displaystyle \Val_{\mathrm{impact}\mid 0}}
  +\underbrace{\tfrac12(b-b_0)\T H^{-1}(b-b_0)}_{\displaystyle \Val_{\mathrm{forecast}\mid H}}
  +\underbrace{\Val_{\mathrm{cross}}}_{\displaystyle -\,(\Theta_{\mathrm{na}}-\Theta_0^H)\T(b-b_0)} .
\]
The forecast--impact cross term is the only generic cross term. It vanishes if and only if
\[
  (\Theta_{\mathrm{na}}-\Theta_0^H)\T(b-b_0)=0,
\]
that is, when the impact correction induced by the restricted forecast is orthogonal, in score space, to the forecast innovation. There is no information--forecast or information--impact cross term because $b_\Inf-b$ has mean zero and enters only through its covariance $\Omega_b$.
\end{theorem}

\begin{proof}
Conditional on $\Inf$, $J_{b_\Inf}$ is a strictly concave quadratic with precision $H$ and maximizer $H^{-1}b_\Inf$. Hence
\[
  \E\Big[J_{b_\Inf}\big(\Theta_{\mathrm{an}}(\Inf)\big)\Big]
  =\tfrac12\E\big[b_\Inf\T H^{-1}b_\Inf\big]
  =\tfrac12 b\T H^{-1}b+\tfrac12\tr(H^{-1}\Omega_b).
\]
Since $\Theta_{\mathrm{na}}=G^{-1}b_0$ is deterministic and $\E[b_\Inf]=b$,
\[
  \E[J_{b_\Inf}(\Theta_{\mathrm{na}})]
  =b\T\Theta_{\mathrm{na}}-\tfrac12\Theta_{\mathrm{na}}\T H\Theta_{\mathrm{na}}.
\]
Subtracting gives
\[
  \Val_3
  =\tfrac12\tr(H^{-1}\Omega_b)
  +\Big[J_b(\Theta_f^H)-J_b(\Theta_{\mathrm{na}})\Big].
\]
Because $\Theta_f^H=H^{-1}b$ maximizes $J_b$, completing the square gives
\[
  J_b(\Theta_f^H)-J_b(\Theta_{\mathrm{na}})
  =\tfrac12\|\Theta_{\mathrm{na}}-\Theta_f^H\|_H^2.
\]
Now write
\[
  \Theta_{\mathrm{na}}-\Theta_f^H
  =\underbrace{(G^{-1}-H^{-1})b_0}_{\Theta_{\mathrm{na}}-\Theta_0^H}
  -\underbrace{H^{-1}(b-b_0)}_{\Theta_f^H-\Theta_0^H}.
\]
Substitution yields the compact formula. Expanding the squared $H$-norm gives
\[
\begin{aligned}
  \tfrac12\|\Theta_{\mathrm{na}}-\Theta_f^H\|_H^2
  &=\tfrac12\|\Theta_{\mathrm{na}}-\Theta_0^H\|_H^2
    +\tfrac12\|H^{-1}(b-b_0)\|_H^2 \\
  &\quad -(\Theta_{\mathrm{na}}-\Theta_0^H)\T H H^{-1}(b-b_0) \\
  &=\tfrac12\|\Theta_{\mathrm{na}}-\Theta_0^H\|_H^2
    +\tfrac12(b-b_0)\T H^{-1}(b-b_0)
    -(\Theta_{\mathrm{na}}-\Theta_0^H)\T(b-b_0),
\end{aligned}
\]
which is the expanded three-way formula. The stated orthogonality condition is exactly the condition under which the cross term is zero. Finally, the signal innovation $b_\Inf-b$ contributes only through $\E[(b_\Inf-b)(b_\Inf-b)\T]=\Omega_b$, so it creates no linear cross term with the deterministic forecast or impact corrections.
\end{proof}

\begin{corollary}[Impact deflates both forecast and information value]\label{cor:threewaydeflation}
Relative to the impact-free dynamic precision $K$, the impact-aware precision is $H=K+2\mathcal L\succeq K$. Hence
\[
  \tfrac12(b-b_0)\T H^{-1}(b-b_0)
  \le
  \tfrac12(b-b_0)\T K^{-1}(b-b_0),
\]
and
\[
  \tfrac12\tr(H^{-1}\Omega_b)
  \le
  \tfrac12\tr(K^{-1}\Omega_b).
\]
Thus impact weakly lowers the marginal value of both forecast and information, strictly in any direction where the forecast innovation or resolved signal covariance loads on an impacted mode.
\end{corollary}
\begin{proof}
Since $\mathcal L\succeq0$, $H\succeq K\succ0$. Operator antitonicity of the inverse gives $H^{-1}\preceq K^{-1}$. Pairing this order with the rank-one matrix $(b-b_0)(b-b_0)\T$ gives the forecast inequality; pairing it with $\Omega_b\succeq0$ gives the information inequality. Strictness follows unless the corresponding vector or covariance is supported entirely on the null directions of $K^{-1}-H^{-1}$.
\end{proof}

\begin{proposition}[Forecast--impact interaction: bounds, angle, and sign]\label{prop:interaction}
In the notation of \Cref{thm:threeway}, set
\[
  a:=\Theta_{\mathrm{na}}-\Theta_0^H=(G^{-1}-H^{-1})b_0,
  \qquad
  v:=\Theta_f^H-\Theta_0^H=H^{-1}(b-b_0),
\]
and define
\[
  A:=\Val_{\mathrm{impact}\mid 0}=\tfrac12\|a\|_H^2,
  \qquad
  F:=\Val_{\mathrm{forecast}\mid H}=\tfrac12\|v\|_H^2.
\]
If $AF>0$, let $\phi$ be the $H$-angle between the impact correction $a$ and the forecast correction $v$:
\[
  \cos\phi:=\frac{a\T H v}{\|a\|_H\|v\|_H}
  =\frac{a\T(b-b_0)}{\|a\|_H\sqrt{(b-b_0)\T H^{-1}(b-b_0)}}.
\]
Then the deterministic part of the three-way value satisfies
\[
  \Val_3-\Val_{\mathrm{info}}
  = A+F-2\sqrt{AF}\cos\phi,
\]
and hence the sharp bounds
\[
  \big(\sqrt A-\sqrt F\big)^2
  \le
  \Val_3-\Val_{\mathrm{info}}
  \le
  \big(\sqrt A+\sqrt F\big)^2.
\]
The forecast and impact faces are \emph{substitutes} when $\cos\phi>0$ because the cross term is negative, \emph{complements} when $\cos\phi<0$ because the cross term is positive, and orthogonal/additive when $\cos\phi=0$. If $A=0$ or $F=0$, the same formulas hold by continuity with zero interaction.
\end{proposition}
\begin{proof}
The cross term in \Cref{thm:threeway} is
\[
  \Val_{\mathrm{cross}}=-a\T(b-b_0)=-a\T H v.
\]
If $AF>0$, the definition of $\phi$ gives $a\T H v=\|a\|_H\|v\|_H\cos\phi=2\sqrt{AF}\cos\phi$, proving the angle formula. Since $-1\le\cos\phi\le1$, the bounds follow immediately. Sharpness occurs when $a$ and $v$ are collinear in the $H$-inner product. The cases $A=0$ or $F=0$ have $a=0$ or $v=0$, so the interaction term is zero and the continuous extension applies.

\end{proof}

\begin{corollary}[Orthogonalized nonnegative split of the forecast--impact interaction]\label{cor:orthosplit}
Use the notation of \Cref{prop:interaction}, and let $v=v_\parallel+v_\perp$ be the $H$-orthogonal decomposition of the forecast correction $v$ into the span of the impact correction $a$ and its $H$-orthogonal complement:
\[
  v_\parallel:=
  \begin{cases}
  \displaystyle \frac{a\T H v}{a\T H a}\,a, & a\ne0,\\[6pt]
  0, & a=0,
  \end{cases}
  \qquad
  v_\perp:=v-v_\parallel .
\]
Then the deterministic part of the three-way value has the nonnegative orthogonal form
\[
  \Val_3-\Val_{\mathrm{info}}
  =\frac12\|a-v_\parallel\|_H^2+\frac12\|v_\perp\|_H^2 .
\]
Consequently the signed cross term is exactly the projection overlap between the forecast correction and the impact correction. The deterministic value is zero iff $v=a$, i.e. iff the forecast correction exactly cancels the price-taking impact bias.
\end{corollary}
\begin{proof}
By \Cref{thm:threeway}, $\Val_3-\Val_{\mathrm{info}}=\tfrac12\|a-v\|_H^2$. Since $a-v_\parallel$ lies in $\operatorname{span}\{a\}$ and $v_\perp$ is $H$-orthogonal to that span,
\[
  \|a-v\|_H^2=\|(a-v_\parallel)-v_\perp\|_H^2
  =\|a-v_\parallel\|_H^2+\|v_\perp\|_H^2.
\]
The zero condition follows because a norm is zero iff $a-v=0$.
\end{proof}

\begin{proposition}[Marginal shadow price of impact]\label{prop:shadowimpact}
Let $H_\tau:=K+2\tau\mathcal L$ for $\tau\ge0$, and write $\delta:=b-b_0$. Define the impact-adjusted marginal forecast and information values
\[
  F(\tau):=\tfrac12\delta\T H_\tau^{-1}\delta,
  \qquad
  I(\tau):=\tfrac12\tr(H_\tau^{-1}\Omega_b).
\]
Then $F$ and $I$ are differentiable and
\[
  F'(\tau)=-\delta\T H_\tau^{-1}\mathcal L H_\tau^{-1}\delta\le0,
  \qquad
  I'(\tau)=-\tr(H_\tau^{-1}\mathcal L H_\tau^{-1}\Omega_b)\le0.
\]
Thus impact weakly lowers the marginal value of forecast and information at every impact intensity, and the marginal loss is largest in modes where $\mathcal L$ overlaps the inverse-precision-filtered forecast innovation or signal covariance.
\end{proposition}
\begin{proof}
Since $dH_\tau/d\tau=2\mathcal L$,
\[
  \frac{d}{d\tau}H_\tau^{-1}=-2H_\tau^{-1}\mathcal L H_\tau^{-1}.
\]
Substitution into $F$ gives the first derivative formula; substitution into $I$ and cyclicity of trace give the second. The matrix $H_\tau^{-1}\mathcal L H_\tau^{-1}$ is positive semidefinite, so the quadratic form and its trace pairing with $\Omega_b\succeq0$ are nonnegative.
\end{proof}

\begin{remark}[What was conjectural is now a theorem]
\Cref{thm:threeway} is the finite-horizon version of \Cref{thm:unified}. The static information--impact theorem is recovered when $T=0$, $b=b_0$, and $\Omega_b=\Omega$. The forecast theorem \Cref{thm:forecastvalue} is recovered when $\mathcal L=0$ and $\Omega_b=0$. The new term is $\Val_{\mathrm{cross}}$: it is the genuine interaction between forecast anticipation and impact anticipation, and it is generally nonzero. Therefore the exact statement is not ``three positive terms'' but ``two positive deterministic terms, one positive information trace, and the unique signed cross term whose sum is a single nonnegative quadratic norm.''
\end{remark}

\subsection{A worked example}

\begin{example}[Two assets, both faces, and the phase transition]\label{ex:worked}
Take $n=2$, $\gamma=1$, $\Sigma=I_2$, $\Lambda=\lambda I_2$, prior mean $\mu=(0.08,\,0.04)\T$, and a signal resolving $\Omega=4\times10^{-4}I_2$. Then $G=(1+\lambda)I$, $H=(1+2\lambda)I$, and $A=\gamma^{-1}\Sigma^{-1}\Lambda=\lambda I$, so $\rho(A)=\lambda$ and the recalibration threshold of \Cref{thm:phase} is at $\lambda=1$.

\emph{Below threshold} ($\lambda=\tfrac25$, $\rho(A)=0.4<1$): naive $\theta_{\mathrm{na}}=\mu/1.4=(\tfrac{2}{35},\tfrac{1}{35})$, impact-aware $\theta_{\mathrm{na}}^{H}=\mu/1.8=(\tfrac{2}{45},\tfrac{1}{45})$. The two values are
\[
  \Val_{\mathrm{impact}}=\tfrac{1+2\lambda}{2}\,\|\theta_{\mathrm{na}}-\theta_{\mathrm{na}}^{H}\|^2=\tfrac{2}{11025}\approx1.814\times10^{-4},
  \qquad
  \Val_{\mathrm{info}}=\tfrac12\tr(H^{-1}\Omega)=\tfrac{1}{4500}\approx2.222\times10^{-4},
\]
summing to $\Val\approx4.036\times10^{-4}$, which equals $\bar J(\theta_{\mathrm{an}})-\bar J(\theta_{\mathrm{na}})$ computed directly. Unit-step recalibration converges at linear rate $0.4$.

\emph{Above threshold} ($\lambda=\tfrac32$, $\rho(A)=1.5>1$): $\theta_{\mathrm{na}}=\mu/2.5=(\tfrac{4}{125},\tfrac{2}{125})$, $\theta_{\mathrm{na}}^{H}=\mu/4=(0.02,0.01)$, with $\Val_{\mathrm{impact}}=\tfrac{9}{25000}=3.6\times10^{-4}$ and $\Val_{\mathrm{info}}=\tfrac1{10000}=10^{-4}$. Unit-step recalibration \emph{diverges}; the damped iteration of \Cref{prop:damped} with $\omega<2/(1+1.5)=0.8$ converges to the same $\theta_{\mathrm{na}}$, but $\theta_{\mathrm{na}}$ is still dominated by $\theta_{\mathrm{na}}^{H}$ by $\Val_{\mathrm{impact}}$. Note the coupling of \Cref{cor:coupling} in the numbers: raising $\lambda$ from $\tfrac25$ to $\tfrac32$ cuts $\Val_{\mathrm{info}}$ from $2.22\times10^{-4}$ to $1.0\times10^{-4}$ --- the same signal is worth less once impact is internalized.
\end{example}

\subsection{From exogenous data to endogenous filtering: the stationary value}\label{sec:stationary}

Two ingredients of the decomposition were taken as exogenous data: the resolved-mean covariance $\Omega$ and the forecast deviation $m-\bar m$. We now close the model, generating the information content from a Kalman filter and lifting the value to the infinite-horizon stationary regime. This replaces ``complete the square against given data'' by the fixed-point equations of estimation and control --- a filtering Riccati for the information face, and a control Riccati with a Lyapunov average for the dynamic value.

\paragraph{A stationary expected-return state.} Let the conditional mean of next-period returns be a latent state $x_t\in\R^n$ (so $\mu_t=x_t$) following a stable Gaussian VAR(1),
\[
  x_{t+1}=\Phi_x x_t+w_{t+1},\qquad w_{t+1}\sim\mathcal N(0,W),\qquad \rho(\Phi_x)<1,
\]
with stationary covariance $\Pi$ solving the discrete Lyapunov equation $\Pi=\Phi_x\Pi\Phi_x\T+W$. An observer with linear--Gaussian measurements $y_t=Cx_t+v_t$, $v_t\sim\mathcal N(0,R)$, forms the steady-state filtered estimate $\hat x_t=\E[x_t\mid\mathcal F_t]$ whose one-step prediction-error covariance $P^-$ is the stabilizing solution of the filtering Riccati equation
\[
  P^-=\Phi_x P^-\Phi_x\T+W-\Phi_x P^-C\T\big(CP^-C\T+R\big)^{-1}CP^-\Phi_x\T,
\]
and whose filtering-error covariance is $P=\big((P^-)^{-1}+C\T R^{-1}C\big)^{-1}$.

\begin{theorem}[The resolved-mean covariance is a Kalman covariance reduction]\label{thm:filtering}
Let a coarse information flow $\mathcal F^{\mathrm c}\subseteq\mathcal F^{\mathrm f}$ be refined to a fine one, both observing the same state $x_t$ through their own linear--Gaussian sensors, with steady-state filtering-error covariances $P^{\mathrm c}\succeq P^{\mathrm f}$. Then the conditional-mean covariance each resolves relative to the unconditional prior is
\[
  \Cov(\hat x^{\bullet})=\Pi-P^{\bullet},\qquad\bullet\in\{\mathrm c,\mathrm f\},
\]
and the incremental resolved-mean covariance of the fine observer over the coarse one is exactly the reduction in steady-state filtering error,
\[
  \Omega\;:=\;\Cov(\hat x^{\mathrm f})-\Cov(\hat x^{\mathrm c})\;=\;P^{\mathrm c}-P^{\mathrm f}\;\succeq\;0 .
\]
Consequently the information term of \Cref{thm:unified,thm:threeway} is endogenous: the trace $\tfrac12\tr(H^{-1}\Omega)$ is the impact-deflated trace of a difference of two Riccati solutions, not a free parameter.
\end{theorem}
\begin{proof}
The filtered estimate is the $L^2$ projection of $x_t$ onto $\mathcal F^{\bullet}_t$, so the error $x_t-\hat x^{\bullet}_t$ is orthogonal to $\hat x^{\bullet}_t$; in stationarity $\Pi=\Cov(\hat x^{\bullet})+\Cov(x_t-\hat x^{\bullet})=\Cov(\hat x^{\bullet})+P^{\bullet}$, giving $\Cov(\hat x^{\bullet})=\Pi-P^{\bullet}$. Subtracting the coarse from the fine case gives $\Omega=(\Pi-P^{\mathrm f})-(\Pi-P^{\mathrm c})=P^{\mathrm c}-P^{\mathrm f}$. Since $\mathcal F^{\mathrm c}\subseteq\mathcal F^{\mathrm f}$, the tower property gives $\Cov(\hat x^{\mathrm f})\succeq\Cov(\hat x^{\mathrm c})$, i.e.\ $P^{\mathrm c}\succeq P^{\mathrm f}$ and $\Omega\succeq0$.
\end{proof}

\begin{corollary}[Comparative statics of endogenous information]\label{cor:filtercompstat}
Sharpening the fine observer's signal --- decreasing its measurement-noise covariance in the Loewner order --- decreases its steady-state error $P^{\mathrm f}$ and hence increases the resolved-mean covariance $\Omega$ and the information value $\tfrac12\tr(H^{-1}\Omega)$, monotonically. Internalizing impact ($\gamma\Sigma\mapsto H=2\Lambda+\gamma\Sigma$) weakly lowers that value, strictly on impacted directions, recovering \Cref{cor:coupling} with an endogenous $\Omega$.
\end{corollary}
\begin{proof}
The stabilizing solution of the filtering Riccati equation is monotone nondecreasing in the noise covariances by the comparison theorem for algebraic Riccati equations \citep{LancasterRodman1995}: a smaller measurement-noise covariance yields a smaller $P^{\mathrm f}$, so $\Omega=P^{\mathrm c}-P^{\mathrm f}$ increases in the Loewner order, and $\tr(H^{-1}\cdot)$ preserves it. The impact-deflation inequality is \Cref{cor:coupling} applied to the endogenous $\Omega$.
\end{proof}

\paragraph{The infinite-horizon stationary value with costs.} Now restore quadratic transaction costs $\Gamma\succeq0$ and discount $\beta\in(0,1)$. An impact-aware trader solves
\[
  \max_{\{\theta_t\}}\ \E\sum_{t\ge0}\beta^t\Big[x_t\T\theta_t-\tfrac12\theta_t\T H\theta_t-\tfrac12(\theta_t-\theta_{t-1})\T\Gamma(\theta_t-\theta_{t-1})\Big],\quad H=2\Lambda+\gamma\Sigma,
\]
while the price-taker solves the same program with $H$ replaced by $G=\Lambda+\gamma\Sigma$ and is then \emph{scored} under the true $H$-objective.

\begin{theorem}[Stationary value of impact anticipation]\label{thm:stationary}
Both programs are discounted linear--quadratic regulators; their optimal policies are partial-adjustment rules $\theta_t=\theta_{t-1}+M_{\bullet}(\mathrm{aim}^{\bullet}_t-\theta_{t-1})$ for $\bullet\in\{G,H\}$, with trade-rate matrices $M_G,M_H$ from the control Riccati equation built on $G$ resp.\ $H$ and forward-looking aims $\mathrm{aim}^{\bullet}_t=\sum_{s\ge0}W^{\bullet}_s\,\Phi_x^{\,s}\hat x_t$. Under the stabilizing Riccati solutions and $\rho(\Phi_x)<1$, the induced closed-loop matrices are Schur. The value of impact anticipation is nonnegative and equals a discounted Lyapunov trace,
\[
  \Val^{\infty}_{\mathrm{impact}}
  =J_H(\pi_H)-J_H(\pi_G)
  =\tfrac12\sum_{t\ge0}\beta^t\,\tr\!\big(\Psi\,S_t\big)\;\ge0,
\]
where $S_t$ is the covariance of the joint state $(\theta_{t-1},\hat x_t)$ under the price-taking closed loop, propagated by $S_{t+1}=F_GS_tF_G\T+\Xi$, and $\Psi\succeq0$ is the $H$-curvature of the value-function gap. In the frictionless limit $\Gamma\to0$ the regulators collapse to the myopic rules $\theta_t=H^{-1}\hat x_t$ and $G^{-1}\hat x_t$. If $S_{\hat x}:=\Cov(\hat x_t)$ denotes the stationary covariance of the forecast state used by the trader, then the value reduces to the stationary average of the static gap of \Cref{thm:impact},
\[
  \Val^{\infty}_{\mathrm{impact}}\xrightarrow[\Gamma\to0]{}
  \frac{1}{1-\beta}\cdot\tfrac12\tr\!\Big(H\,(G^{-1}-H^{-1})\,S_{\hat x}\,(G^{-1}-H^{-1})^{\!\top}\Big).
\]
Under full observation $S_{\hat x}=\Pi$; under Kalman filtering $S_{\hat x}=\Pi-P$. Separately, the static price-taking recalibration map
\[
  \theta^{(k+1)}=\gamma^{-1}\Sigma^{-1}(\mu-\Lambda\theta^{(k)})
\]
is globally stable iff $\rho(\gamma^{-1}\Sigma^{-1}\Lambda)<1$, the threshold of \Cref{thm:phase}. This recalibration threshold governs an algorithmic fixed-point iteration; it is not the closed-loop stability condition of the already-solved dynamic trading regulator.
\end{theorem}
\begin{proof}
Throughout, $\bullet\in\{G,H\}$ indexes the holding-precision matrix of the two programs ($G=\Lambda+\gamma\Sigma$ for the price-taker, $H=2\Lambda+\gamma\Sigma$ for the impact-aware trader); both are symmetric positive definite since $\Sigma\succ0$, $\gamma>0$, $\Lambda\succeq0$. We prove the four assertions in turn: the linear--quadratic reduction and partial-adjustment form; the Lyapunov-trace identity together with nonnegativity; the frictionless limit; and the separation between dynamic closed-loop stability and static recalibration stability.

\smallskip\noindent\textit{Step 1 (reduction to a homogeneous LQ regulator).}
Take the joint state $z_t:=(\theta_{t-1},\hat x_t)\in\R^{2n}$ and control $u_t:=\theta_t$. The trader is $\mathcal F_t$-adapted, so $\theta_t$ is $\mathcal F_t$-measurable and the only place the latent return $x_t$ enters the objective, the bilinear term $x_t\T\theta_t$, satisfies $\E[x_t\T\theta_t\mid\mathcal F_t]=\hat x_t\T\theta_t$. The expected per-period reward is therefore exactly the full-information reward with $x_t$ replaced by the filtered state $\hat x_t$ --- the separation principle holds with no residual constant, because the objective is \emph{linear}, not quadratic, in $x_t$. The filtered state is the stationary steady-state Kalman process
\[
  \hat x_{t+1}=\Phi_x\hat x_t+\iota_{t+1},\qquad \E[\iota_{t+1}\iota_{t+1}\T]=\Sigma_\iota,\quad \iota_{t+1}\perp\mathcal F_t,
\]
with $\Cov(\hat x_t)=\Pi-P$ by \Cref{thm:filtering} (and $\hat x_t=x_t$, $\Sigma_\iota=W$, $\Cov(\hat x_t)=\Pi$ under full observation). Collecting terms, the dynamics and reward are
\[
  z_{t+1}=\mathcal A z_t+\mathcal B u_t+\zeta_{t+1},\quad
  \mathcal A=\begin{psmallmatrix}0&0\\[1pt]0&\Phi_x\end{psmallmatrix},\ \
  \mathcal B=\begin{psmallmatrix}I\\[1pt]0\end{psmallmatrix},\ \
  \zeta_{t+1}=\begin{psmallmatrix}0\\[1pt]\iota_{t+1}\end{psmallmatrix},\ \
  \Xi:=\E[\zeta\zeta\T]=\begin{psmallmatrix}0&0\\[1pt]0&\Sigma_\iota\end{psmallmatrix},
\]
\[
  r_{\bullet}(z_t,u_t)=\hat x_t\T u_t-\tfrac12u_t\T\bullet\,u_t-\tfrac12(u_t-\theta_{t-1})\T\Gamma(u_t-\theta_{t-1}),
\]
which is jointly quadratic in $(z_t,u_t)$ with no constant: a homogeneous discounted LQ regulator. Since $\rho(\Phi_x)<1$ and $\beta\in(0,1)$, the pair $(\sqrt\beta\,\mathcal A,\mathcal B)$ is stabilizable (the $\hat x$-block is already Schur and the $\theta$-block is reachable through $\mathcal B$) and the relevant quadratic pencil is detectable; by the standard theory of the discrete algebraic Riccati equation \citep{LancasterRodman1995} each program admits a unique stabilizing symmetric solution $\mathcal K_\bullet=\mathcal K_\bullet\T$. Under the reward-maximization sign convention used here,
\[
  V_\bullet(z)=-\tfrac12z\T\mathcal K_\bullet z+c_\bullet,
\]
$\mathcal K_\bullet$ is generally indefinite rather than positive semidefinite: even the frictionless one-period value is positive quadratic in $\hat x_t$. Concavity of the Bellman maximization is governed instead by the positive control curvature below. The linear part of the value vanishes because $\E\zeta=0$, so the only inhomogeneity feeds the constant $c_\bullet$, and the optimal feedback is homogeneous and linear:
\[
  u=\pi_\bullet(z)=K_\bullet z,\qquad
  K_\bullet=(\Psi^{\mathrm{ctrl}}_\bullet)^{-1}L_\bullet,\qquad
  \Psi^{\mathrm{ctrl}}_\bullet:=\bullet+\Gamma+\beta\,\mathcal B\T\mathcal K_\bullet\mathcal B\succ0,
\]
where $L_\bullet$ is the $u$-versus-$z$ cross-block of the reward-plus-continuation form and $\Psi^{\mathrm{ctrl}}_\bullet$ is its (state-independent) control curvature. Writing the gain in its two $n\times n$ blocks, $K_\bullet=[\,I-M_\bullet\ \ \ M_\bullet A^{\mathrm{aim}}_\bullet\,]$, the feedback rearranges to the partial-adjustment rule $\theta_t=\theta_{t-1}+M_\bullet(\operatorname{aim}^\bullet_t-\theta_{t-1})$ with trade-rate $M_\bullet$ a function of $(\bullet,\Gamma,\beta,\mathcal K_\bullet)$ and aim $\operatorname{aim}^\bullet_t=A^{\mathrm{aim}}_\bullet\hat x_t=\sum_{s\ge0}W^\bullet_s\Phi_x^{\,s}\hat x_t$ a discounted sum of expected future Markowitz targets $\bullet^{-1}\E_t[x_{t+s}]=\bullet^{-1}\Phi_x^{\,s}\hat x_t$ at precision $\bullet$; this is exactly the structural solution of \Cref{thm:gp}/\citet{GarleanuPedersen2013} for the present coefficients, with the discount absorbed by the customary $\sqrt\beta$ rescaling. This proves the first assertion.

\smallskip\noindent\textit{Step 2 (the value of impact anticipation is a discounted Lyapunov trace).}
Let $Q_H(z,u):=r_H(z,u)+\beta\,\E_\zeta V_H(\mathcal A z+\mathcal B u+\zeta)$ be the $H$-state--action value. Its $u$-Hessian is $\partial^2_{uu}Q_H=-(H+\Gamma)+\beta\mathcal B\T(-\mathcal K_H)\mathcal B=-\Psi^{\mathrm{ctrl}}_H$, \emph{independent of $z$}, and the maximum over $u$ is attained at $u=\pi_H(z)=K_Hz$ with $V_H(z)=Q_H(z,K_Hz)$ (the Bellman identity). Hence for the price-taking control $u=\pi_G(z)=K_Gz$, exact second-order expansion of the concave quadratic $Q_H(z,\cdot)$ about its maximizer gives the one-step advantage gap
\begin{equation}\label{eq:advgap}
  \Delta_H(z):=V_H(z)-Q_H(z,K_Gz)
  =\tfrac12\,(K_Gz-K_Hz)\T\Psi^{\mathrm{ctrl}}_H(K_Gz-K_Hz)
  =\tfrac12\,z\T\Psi\,z\ \ge 0,
\end{equation}
where we have set the joint-state curvature
\[
  \boxed{\;\Psi:=(K_G-K_H)\T\,\Psi^{\mathrm{ctrl}}_H\,(K_G-K_H)\;=\;(K_G-K_H)\T\big(H+\Gamma+\beta\,\mathcal B\T\mathcal K_H\mathcal B\big)(K_G-K_H)\;\succeq0\;}
\]
--- the $H$-value control curvature pulled back through the gain gap $K_G-K_H$. Now let $g(z):=V_H(z)-J_H(\pi_G;z)$, where $J_H(\pi_G;z)$ is the discounted $H$-value started at $z$ under the fixed linear policy $\pi_G$, which obeys the policy-evaluation identity $J_H(\pi_G;z)=r_H(z,K_Gz)+\beta\E_\zeta J_H(\pi_G;F_Gz+\zeta)$ with closed loop $F_G:=\mathcal A+\mathcal B K_G$. Because $Q_H(z,K_Gz)=r_H(z,K_Gz)+\beta\E_\zeta V_H(F_Gz+\zeta)$, subtracting the two displays yields the fixed-point relation
\[
  g(z)=\Delta_H(z)+\beta\,\E_\zeta\,g(F_Gz+\zeta).
\]
Iterating $N$ times along the price-taking trajectory $z_{t+1}=F_Gz_t+\zeta_{t+1}$,
\[
  g(z_0)=\sum_{t=0}^{N-1}\beta^t\,\E\!\left[\Delta_H(z_t)\mid z_0\right]+\beta^N\,\E\!\left[g(z_N)\mid z_0\right].
\]
Both $V_H$ and $J_H(\pi_G;\cdot)$ are quadratics and, by Step~4 below, $F_G$ is Schur, so $\E[g(z_N)]$ grows at most linearly in $N$ while $\beta^N\to0$; the remainder vanishes. Taking expectations over the initial second moment $S_0:=\E[z_0z_0\T]$ and inserting \eqref{eq:advgap},
\[
  \Val^{\infty}_{\mathrm{impact}}=J_H(\pi_H)-J_H(\pi_G)=\E[g(z_0)]
  =\tfrac12\sum_{t\ge0}\beta^t\,\tr\!\big(\Psi\,S_t\big),
  \qquad
  S_{t+1}=F_GS_tF_G\T+\Xi,
\]
the recursion for $S_t=\E[z_tz_t\T]$ following from $z_{t+1}=F_Gz_t+\zeta_{t+1}$ with $\zeta_{t+1}\perp z_t$ and $\E\zeta_{t+1}=0$. This is the stated identity. Nonnegativity is immediate from $\Psi\succeq0$ and $S_t\succeq0$, and also follows directly because $\pi_H$ maximizes $J_H$ over admissible (stabilizing) policies while $\pi_G$ is admissible. This proves the second assertion.

\smallskip\noindent\textit{Step 3 (frictionless limit).}
Set $\Gamma=0$. The trading cost was the only term coupling $\theta_t$ to $\theta_{t-1}$, so the program decouples across periods: $V_\bullet(z)=\max_u\{\hat x\T u-\tfrac12u\T\bullet\,u\}+\beta\,\E_\zeta V_\bullet(z^+)$, whose maximizer is the one-period Markowitz target $\theta_t=\bullet^{-1}\hat x_t$ (equivalently $M_\bullet\to I$ and only the $s=0$ aim weight survives, $\operatorname{aim}^\bullet_t\to\bullet^{-1}\hat x_t$). The value Hessian then has a zero $\theta_{t-1}$-block, so $V_\bullet$ does not depend on its first argument; consequently the control enters the continuation only through that vanishing block and $\mathcal B\T\mathcal K_H\mathcal B=0$, giving $\Psi^{\mathrm{ctrl}}_H=H$ exactly. The gain gap reduces to $K_G-K_H=[\,0\ \ \ G^{-1}-H^{-1}\,]$, so by \eqref{eq:advgap}
\[
  \Delta_H(z_t)=\tfrac12\,\hat x_t\T(G^{-1}-H^{-1})\T H\,(G^{-1}-H^{-1})\,\hat x_t,
\]
which is precisely the static value gap of \Cref{thm:impact} with $\mu$ replaced by the realized forecast $\hat x_t$ (there $\thetana-\thetaan=(G^{-1}-H^{-1})\mu$ and the gap is $\tfrac12\norm{\thetana-\thetaan}_H^2$). Taking expectations with $S_{\hat x}:=\Cov(\hat x_t)=\E[\hat x_t\hat x_t\T]$ and using cyclicity and the symmetry of $G^{-1}-H^{-1}$, each period contributes the constant $\tfrac12\tr\!\big(H(G^{-1}-H^{-1})S_{\hat x}(G^{-1}-H^{-1})\T\big)$; summing the geometric series $\sum_{t\ge0}\beta^t$ gives
\[
  \Val^{\infty}_{\mathrm{impact}}\xrightarrow[\Gamma\to0]{}\frac{1}{1-\beta}\cdot\tfrac12\tr\!\Big(H\,(G^{-1}-H^{-1})\,S_{\hat x}\,(G^{-1}-H^{-1})^{\!\top}\Big),
\]
with $S_{\hat x}=\Pi$ under full observation and $S_{\hat x}=\Pi-P$ under Kalman filtering by \Cref{thm:filtering}. This proves the third assertion.

\smallskip\noindent\textit{Step 4 (closed-loop stability versus recalibration stability).}
The discounted Lyapunov sum $\sum_t\beta^t S_t$ converges, and the remainder in Step~2 vanishes, under the stabilizing price-taking closed loop $F_G=\mathcal A+\mathcal B K_G$. This is the dynamic trading stability object: $F_G$ maps the joint state $(\theta_{t-1},\hat x_t)$ under the already-solved optimal price-taking policy. The stabilizing Riccati solution in Step~1 gives $\rho(F_G)<1$ under the stated hypotheses, so $S_t$ admits the usual stationary Lyapunov recursion and the performance-difference sum is well defined.

This should be distinguished from the static price-taking recalibration map of \Cref{thm:phase}. That map is an algorithm for reaching a fixed point of the impacted expected-return law by repeatedly re-estimating returns at the current book and then re-optimizing:
\[
  \theta^{(k+1)}=\gamma^{-1}\Sigma^{-1}(\mu-\Lambda\theta^{(k)}),
  \qquad D\mathcal R=-\gamma^{-1}\Sigma^{-1}\Lambda .
\]
It is globally convergent from every initialization iff
\[
  \rho(\gamma^{-1}\Sigma^{-1}\Lambda)<1,
\]
which is the threshold of \Cref{thm:phase}. In the frictionless dynamic regulator $\Gamma=0$, the trading policy itself simply reaches the current target in one step, $\theta_t=G^{-1}\hat x_t$ for the price-taker and $\theta_t=H^{-1}\hat x_t$ for the impact-aware trader; its closed-loop stability is inherited from the forecast/filter state $\hat x_t$. The threshold therefore governs iterative self-calibration of a price-taking estimator, not the physical closed-loop stability of the solved trading policy. This proves the final assertion.

The closed forms for $M_\bullet,\,\Psi,\,F_\bullet$ are the stabilizing Riccati and Lyapunov solutions just exhibited.
\end{proof}

\begin{remark}[What the endogenization buys]
\Cref{thm:filtering,thm:stationary} answer the two exogeneity objections to the decomposition. The information term is no longer a free covariance $\Omega$ but the gap between two Kalman filters, monotone in signal quality through Riccati monotonicity; the dynamic value is no longer a single completed square but a discounted Lyapunov trace coupling the control Riccati to the stationary state law. The three-way structure of \Cref{thm:threeway} survives --- information enters through $\Omega=P^{\mathrm c}-P^{\mathrm f}$, impact through the $G\!\to\!H$ precision shift, forecast through the predictable path $\Phi_x^{\,s}\hat x_t$ feeding the aim --- but every ingredient is now the solution of an estimation or control fixed-point equation rather than given data.
\end{remark}

\paragraph{Status.} \Cref{thm:unified,thm:threeway,cor:coupling,cor:threewaydeflation,prop:interaction,cor:orthosplit,cor:plugin,cor:hpplugin,prop:curvature,prop:mvinfo,prop:damped,thm:filtering,cor:filtercompstat,thm:stationary} and \Cref{ex:worked} are proved in the stated settings. The finite-horizon information--forecast--impact decomposition is a theorem; the forecast--impact cross term is bounded, classified, and orthogonalized; the information face is endogenized as a Kalman covariance reduction; and the value is carried to the infinite-horizon stationary regime through a complete linear--quadratic argument --- separation, the stabilizing control Riccati, a performance-difference advantage gap, and a closed-loop Lyapunov recursion.

\section{Numerical Illustration: The Impact Threshold}\label{sec:numerics}

The impact and phase-transition theorems are fully analytic, but a numerical illustration makes the mechanism transparent. We simulate the recalibration map $\theta^{(k+1)}=\gamma^{-1}\Sigma^{-1}(\mu-\Lambda\theta^{(k)})$ across values of $\rho(\gamma^{-1}\Sigma^{-1}\Lambda)$, comparing $\thetaM=\gamma^{-1}\Sigma^{-1}\mu$, $\thetana=(\Lambda+\gamma\Sigma)^{-1}\mu$, and $\thetaan=(2\Lambda+\gamma\Sigma)^{-1}\mu$.

\paragraph{Design.}
We fix $n=20$ and $\gamma=3$, draw a positive-definite covariance $\Sigma$ with eigenvalues spread over $[0.5,2]$ in a random orthonormal basis, and an expected-return vector $\mu$ with entries in $[0.05,0.15]$. We construct a family of impact matrices
\[
  \Lambda_c=c\,\Sigma^{1/2}R\,\Sigma^{1/2},
\]
where $R\succeq0$ has a spread spectrum, so that $A_c=\gamma^{-1}\Sigma^{-1}\Lambda_c$ is similar to a symmetric positive-semidefinite matrix and $\rho(A_c)$ increases monotonically with the impact intensity $c$. The impact base is normalized so that $\rho(\gamma^{-1}\Sigma^{-1}\Lambda_1)=\tfrac13$, hence $\rho(A_c)=c/3$, and we pick $c\in\{1.05,3.00,4.80\}$ to place $\rho(A_c)$ exactly at $\{0.35,1.00,1.60\}$: below, at, and above the critical surface.

\paragraph{Results.}
\Cref{tab:phase} reports, for each regime, the impact scale, the spectral radius, the behavior of the unit-step recalibration started at $\theta^{(0)}=0$, and the realized objectives $J(\thetana)$ and $J(\thetaan)$. The anticipation gap $J(\thetaan)-J(\thetana)$ is positive in every regime and \emph{grows} with impact, from $7.6\times10^{-4}$ at $\rho(A)=0.35$ to $2.4\times10^{-3}$ at $\rho(A)=1.6$. \Cref{fig:phase} plots the recalibration error $\|\theta^{(k)}-\thetana\|_2$: geometric convergence at rate $\rho(A)=0.35$ below threshold, a marginal plateau at $\rho(A)=1.0$, and exponential divergence at $\rho(A)=1.6$.

\begin{table}[H]
\centering
\caption{Numerical illustration of the impact threshold ($n=20$, $\gamma=3$). The price-taking fixed point $\thetana$ exists in every regime as the solution of a linear system; the unit-step recalibration that would reach it is globally convergent only below the threshold.}
\label{tab:phase}
\begin{tabular}{@{}lccccc@{}}
\toprule
Regime & Impact $c$ & $\rho(\gamma^{-1}\Sigma^{-1}\Lambda_c)$ & Naive recalibration & $J(\thetana)$ & $J(\thetaan)$ \\
\midrule
Low      & $1.05$ & $0.35$ & converges (rate $0.35$)      & $0.02160$ & $0.02236$ \\
Critical & $3.00$ & $1.00$ & marginal --- no convergence  & $0.01290$ & $0.01488$ \\
High     & $4.80$ & $1.60$ & diverges                     & $0.00911$ & $0.01153$ \\
\bottomrule
\end{tabular}
\end{table}

\begin{figure}[H]
\centering
\includegraphics[width=0.72\textwidth]{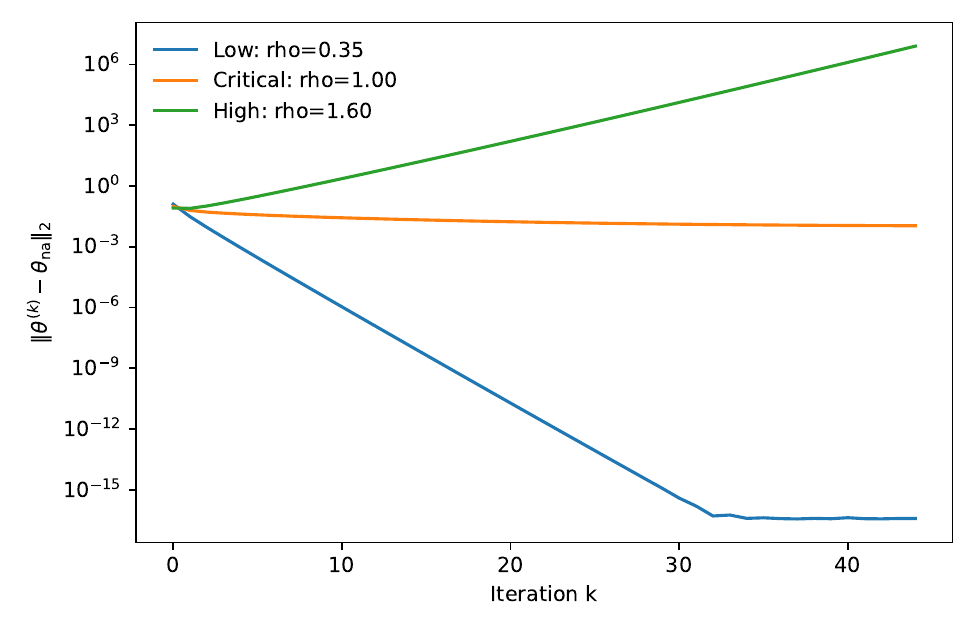}
\caption{Recalibration error $\|\theta^{(k)}-\thetana\|_2$ versus iteration $k$ for the three regimes of \Cref{tab:phase}. Below the threshold ($\rho(A)=0.35$) the unit-step map converges geometrically; at the threshold ($\rho(A)=1.0$) it neither converges nor blows up; above it ($\rho(A)=1.6$) it diverges. The anticipatory optimum $\thetaan$ requires no iteration in any regime.}
\label{fig:phase}
\end{figure}

\paragraph{Interpretation.}
The experiment separates two ideas that are often conflated. The price-taking fixed point can exist even when the learning dynamics used to reach it are unstable. The anticipatory optimizer does not chase a fixed point of the induced return law; it maximizes the deployed objective directly. This is why impact anticipation is not merely a more conservative Markowitz allocation, but a change from estimation to control --- and why the value of that change, $J(\thetaan)-J(\thetana)$, is largest exactly in the high-impact regime where naive recalibration fails.

\subsection{Verifying the three-way decomposition and the estimation penalty}\label{sec:numthreeway}

The phase experiment above isolates impact. The stacked LQG diagnostic below displays the three-way decomposition of \Cref{thm:threeway}, its signed cross term, and the estimation penalty of \Cref{thm:estpenalty} on a two-period, two-asset instance ($T=2$, $n=2$, $\beta=0.97$).

\paragraph{The cross term is not negligible.}
\Cref{tab:threeway} reports the four components of the deterministic-plus-information value in two configurations that share the \emph{same} information trace, impact gap $A$, and forecast premium $F$, and differ only in the $H$-angle between the forecast and impact corrections. When they are nearly aligned ($\cos\phi\approx+0.99$, substitutes) the negative cross term cancels most of the value; when nearly anti-aligned ($\cos\phi\approx-0.99$, complements) the positive cross term roughly doubles it. The Monte-Carlo value over the signal matches the closed-form decomposition to sampling error in both cases, and the orthogonalized split of \Cref{cor:orthosplit} reproduces the deterministic part exactly. Reporting only ``information $+$ impact $+$ forecast'' would misstate the total by the cross term --- here a factor of about $2.1$ between the two configurations.

\begin{table}[H]
\centering
\caption{End-to-end check of the three-way decomposition (\Cref{thm:threeway}) on a stacked LQG instance, values in units of $10^{-4}$. Both columns share the same information, impact, and forecast components and differ only in the sign of the forecast--impact cross term; the closed-form total matches the direct Monte-Carlo value (over the signal) to sampling error.}
\label{tab:threeway}
\small
\begin{tabularx}{\textwidth}{@{}Ycc@{}}
\toprule
Component & Substitutes ($\cos\phi\approx+0.99$) & Complements ($\cos\phi\approx-0.99$) \\
\midrule
$\Val_{\mathrm{info}}=\tfrac12\tr(H^{-1}\Omega_b)$        & $2.237$  & $2.237$ \\
$A=\Val_{\mathrm{impact}\mid 0}$                          & $0.639$  & $0.639$ \\
$F=\Val_{\mathrm{forecast}\mid H}$                        & $0.633$  & $0.633$ \\
$\Val_{\mathrm{cross}}=-a\T(b-b_0)$                       & $-1.265$ & $+1.265$ \\
\midrule
Total (closed-form decomposition)                        & $2.244$  & $4.774$ \\
Total (direct Monte Carlo over signal)                   & $2.242$  & $4.774$ \\
\bottomrule
\end{tabularx}
\end{table}

\paragraph{Estimation can make anticipation worth less than nothing.}
For a six-asset deployed precision with true score $b$ and estimation noise $\Sigma_\eps$, the oracle anticipatory value over the zero baseline is $\tfrac12 q$ with $q=b\T H^{-1}b$. In a high-noise regime ($p=\tr(H^{-1}\Sigma_\eps)\approx1.67$ against $q\approx0.12$), the \emph{unshrunk} plug-in's expected value is $\tfrac12(q-p)\approx-0.77$: anticipating an over-trusted estimate is strictly worse than not anticipating at all, the exact average-case form of \Cref{prop:value}'s second clause. The reliability-ratio shrinkage of \Cref{cor:shrinkage}, $s^\star=q/(q+p)\approx0.069$, restores a nonnegative value $\tfrac12 q^2/(q+p)\approx+4.2\times10^{-3}$, and the Bayes matrix shrinkage of \Cref{rem:bayesshrink} dominates the scalar rule while matching its closed form $\tfrac12\tr(M^\star\Sigma_b)$. The measured plug-in suboptimality relative to the oracle equals $\tfrac12\tr(H^{-1}\Sigma_\eps)$ to Monte-Carlo error, confirming the estimation--information mirror.

\section{Empirical Protocol and Falsifiability}\label{sec:protocol}

A strong empirical version of the paper should not merely report that an anticipatory allocation looks different from a naive one. It should show that the enrichment is decision-relevant, stable, and not dominated by estimation error. The following protocol turns the theory into testable diagnostics.

\paragraph{Objects to estimate.}
A rolling implementation needs five inputs: a baseline mean--covariance pair $(\widehat\mu_t,\widehat\Sigma_t)$; a forecast stack $\widehat m_t=(\widehat m_{t,0},\ldots,\widehat m_{t,T})$ or factor law $(\widehat B_t,\widehat\Phi_t)$; an impact operator $\widehat{\mathcal L}_t$ estimated from signed flow, order imbalance, or execution slippage; a signal-resolution matrix $\widehat\Omega_{b,t}$ when private or alternative information is used; and an ambiguity radius $\widehat\eps_t$ for the robust layer. The forecast and impact objects must be estimated on information available before deployment, then frozen during the out-of-sample decision period.

\begin{center}
\small
\begin{tabularx}{\textwidth}{@{}L{0.19\textwidth}L{0.28\textwidth}Y@{}}
\toprule
Face & Diagnostic statistic & Falsification / failure mode \\
\midrule
Information & $\tr(\widehat H_t^{-1}\widehat\Omega_{b,t})$ or an information-drift estimate & The signal is economically vacuous if the trace is statistically indistinguishable from zero after costs and turnover. \\
Forecast & $(\widehat b_t-\widehat b_{0,t})\T\widehat H_t^{-1}(\widehat b_t-\widehat b_{0,t})$ & Predictive accuracy is not enough; by \Cref{cor:decisionforecast}, the forecast fails if errors live in directions the optimizer cannot express after risk, impact, and transaction costs. \\
Impact & $\widehat{\mathcal L}_t\widehat\Theta^{\mathrm M}_t$ and $\rho(\widehat K_t^{-1}\widehat{\mathcal L}_t)$ & Impact anticipation is vacuous if the Markowitz direction lies in the impact null space; unit-step naive recalibration is unstable when the spectral radius exceeds one. \\
Interaction & $\cos\widehat\phi_t$ from \Cref{prop:interaction} & A negative additive claim is false if forecast and impact are substitutes; the cross term must be reported, not hidden. \\
Robustness & $\widehat\Delta_t-2\widehat\eta_t$ or the robust objective gap & Plug-in anticipation should be rejected when the estimated value is smaller than the model-error band of \Cref{cor:plugin,cor:hpplugin}; net deployment should also subtract turnover, execution, and data costs. \\
Estimation penalty & $\tfrac12\tr(\widehat H_t^{-1}\widehat\Sigma_{\eps,t})$ against the structural value & By \Cref{thm:estpenalty}, plug-in anticipation is expected to lose value once the estimation-noise trace exceeds the resolved-structure trace; the fix is to deploy the reliability-ratio or Bayes-shrunk score of \Cref{cor:shrinkage,rem:bayesshrink} rather than the raw estimate. \\
\bottomrule
\end{tabularx}
\end{center}

\paragraph{Rolling implementation.}
At each rebalance date $t$, build
\[
  \widehat K_t=\widehat{\mathcal R}_t+D\T\widehat{\mathcal C}_tD,
  \qquad
  \widehat H_t=\widehat K_t+2\widehat{\mathcal L}_t,
  \qquad
  \widehat G_t=\widehat K_t+\widehat{\mathcal L}_t,
\]
then compute the restricted price-taking path, the restricted impact-aware path, and the full anticipatory path:
\[
  \widehat\Theta_{\mathrm{na},t}=\widehat G_t^{-1}\widehat b_{0,t},
  \qquad
  \widehat\Theta_{0,t}^{H}=\widehat H_t^{-1}\widehat b_{0,t},
  \qquad
  \widehat\Theta_{\mathrm{an},t}=\widehat H_t^{-1}\widehat b_{\Inf,t}.
\]
A robust deployment replaces the last solve by the convex program
\[
  \widehat\Theta^{\mathrm{rob}}_t
  \in\argmax_{\Theta\in\Theta_t}
  \left\{\widehat b_{\Inf,t}\T\Theta
  -\tfrac12\Theta\T\widehat H_t\Theta
  -\widehat\eps_t\|\Theta\|_*\right\}.
\]
This is the empirical version of the paper's message: enrich the model, but price the uncertainty of the enrichment.

\paragraph{Rolling out-of-sample evaluation.}
Hold the chosen positions over the next evaluation interval and compute realized utilities using the same accounting convention for all policies:
\[
  \widehat U_{t+1}(\Theta_t)
  :=r_{t+1}\T\theta_{t}
  -\tfrac12(\theta_t-\theta_{t-1})\T\Gamma_t(\theta_t-\theta_{t-1})
  -\widehat c_t(\Theta_t)
  -\tfrac\gamma2\theta_t\T\widehat\Sigma_t\theta_t,
\]
where $\widehat c_t$ is the realized or estimated execution/impact cost and $\theta_t$ is the first deployed block of the stacked decision $\Theta_t$. The empirical value is the paired difference
\[
  \widehat\Val^{\mathrm{oos}}
  =\frac1N\sum_{t=1}^N
  \big[\widehat U_{t+1}(\widehat\Theta_{\mathrm{an},t})-\widehat U_{t+1}(\widehat\Theta_{\mathrm{na},t})\big],
\]
and the robust value replaces $\widehat\Theta_{\mathrm{an},t}$ by $\widehat\Theta^{\mathrm{rob}}_t$. Paired block bootstrap intervals are natural because the same return path evaluates both policies and serial dependence is expected.

\paragraph{Minimum reporting standard.}
An empirical paper should report (i) the raw and cost-adjusted value gap; (ii) turnover and concentration of each policy; (iii) the spectral radius $\rho(\widehat K^{-1}\widehat{\mathcal L})$ through time; (iv) the information trace, forecast premium, impact gap, forecast--impact interaction, and orthogonal projection components separately; (v) the realized sign of the plug-in condition $\widehat\Delta-2\widehat\eta$ and, when confidence bands are available, the high-probability condition $\widehat\Delta-2\widehat\eta_\delta$; and (vi) a robustness sweep over $\eps$. A positive in-sample anticipatory value with a negative out-of-sample robust value is not evidence of anticipation; it is evidence of overfitting the future model.

\section{Synthesis}\label{sec:synth}

\subsection{The three values of anticipation}

\begin{center}
\begin{tabularx}{\textwidth}{@{}L{0.18\textwidth}L{0.30\textwidth}Y@{}}
\toprule
Face & What is anticipated & Value of anticipation $\Val$ \\
\midrule
Information (\Cref{sec:info}) & a signal $L$ of the future under an enlarged filtration & information-drift energy $\tfrac12\E\!\int_0^T\alpha_t^2\,dt$; a mutual-information / relative-entropy quantity under standard enlargement assumptions ($=\Ent(L)$ if $L$ discrete; $=0$ iff $\alpha\equiv0$) (\Cref{thm:pk}) \\
Forecast (\Cref{sec:forecast}) & predicted return dynamics over a trading horizon & multi-period value minus restricted-forecast value; exact quadratic premium $\tfrac12(m-\bar m)\T\mathcal B K^{-1}\mathcal B(m-\bar m)$ and forward-looking aim portfolio (\Cref{thm:forecastvalue,thm:gp}) \\
Impact (\Cref{sec:impact}) & the return shift the strategy itself causes & performative gap $\tfrac12(\thetana{-}\thetaan)\T(2\Lambda{+}\gamma\Sigma)(\thetana{-}\thetaan)$ (\Cref{thm:impact}) \\
Robustness (\Cref{sec:robust}) & adversarial misspecification of the anticipated model & dual-norm penalty $\eps\|\theta\|_*$ capping optimizer overreach under Wasserstein ambiguity (\Cref{prop:dro}) \\
\bottomrule
\end{tabularx}
\end{center}

The finite-horizon LQG theorem (\Cref{thm:threeway}) unifies the first three faces in one formula:
\[
  \Val_3
  =\tfrac12\tr(H^{-1}\Omega_b)
  +\tfrac12\|(G^{-1}-H^{-1})b_0-H^{-1}(b-b_0)\|_H^2.
\]
Expanding the squared norm gives an information trace, an impact gap, a forecast premium, and the signed forecast--impact cross term $-(\Theta_{\mathrm{na}}-\Theta_0^H)\T(b-b_0)$. \Cref{prop:interaction} shows that this term is not an ambiguity in the theory: it is exactly the $H$-angle between the impact correction and the forecast correction, bounded between the substitute and complement extremes. The static information--impact formula is the special case $b=b_0$, and the pure forecast formula is the special case $\mathcal L=0$ and $\Omega_b=0$.

\subsection{When anticipation helps, is vacuous, or harms}

The common theorem behind all three faces is simple. If the enriched model is correctly specified and the naive policy remains feasible, then $\Val=J(\thetaan)-J(\thetana)\ge0$. Anticipation is therefore valuable not because it is optimistic, but because it solves a strictly richer decision problem.

Anticipation is \emph{vacuous} when the enrichment has no decision-relevant content. In the information case this means $\alpha\equiv0$. In the forecast case it means predictable expected returns do not change the optimal trading path relative to the myopic portfolio. In the impact case it means $\Lambda\thetaM=0$, so the Markowitz direction lies in the null space of impact.

Anticipation becomes harmful only when the enrichment is estimated and then over-optimized. The optimizer can exploit noise in a signal, a forecast, or an impact matrix exactly because the enriched model gives it more directions in which to act; the realized downside is bounded by twice the model-error sup (\Cref{prop:value}), operationalized by the deterministic and high-probability plug-in hurdles of \Cref{cor:plugin,cor:hpplugin}, and capped by the robust regularizer of \Cref{prop:dro}. In the linear--quadratic case the harm is moreover \emph{exact}: by \Cref{thm:estpenalty} the expected value of estimated anticipation is the structural value minus an estimation penalty $\tfrac12\tr(H^{-1}\Sigma_\eps)$ that mirrors the information gain $\tfrac12\tr(H^{-1}\Omega)$ term for term, so anticipation pays in expectation iff resolved structure beats injected estimation noise in the precision trace, and the bias--variance-optimal response is to shrink the score by its reliability ratio (\Cref{cor:shrinkage,rem:bayesshrink}). This is why robust anticipation is not an optional refinement but the natural counterpart of estimated anticipation: estimated anticipation $\Rightarrow$ regularized or distributionally robust anticipation, the worst-case face of the same shrinkage the expectation already prescribes. In one sentence: anticipation is the move from estimation to control; it creates value when the enriched model is true, has no effect when the enrichment is irrelevant, and destroys value only when model error is optimized as if it were structure.

\subsection{Open problems}
\begin{enumerate}[leftmargin=2em,itemsep=3pt]
\item \textbf{Progressive information and the information rate.} The static noisy case is the mean--variance trace $\tfrac1{2\gamma}\tr(\Sigma^{-1}\Omega)$ (\Cref{prop:mvinfo}), and the stationary noisy case is now endogenized as a steady-state Kalman covariance reduction $\Omega=P^{\mathrm c}-P^{\mathrm f}$ (\Cref{thm:filtering}). What remains is the genuinely progressive enlargement --- replacing terminal mutual information by an information-rate decomposition along the filtration --- and the reconciliation of the log-utility entropy and mean--variance trace functionals (\Cref{rem:tracevslogdet}) under a common information-rate object.
\item \textbf{Decision-focused forecasting.} Estimate the forecast objects $(B,\Phi)$ through the trading objective itself, obtaining finite-sample guarantees on realized portfolio value rather than on prediction error.
\item \textbf{Transient-impact anticipation.} Replace the permanent impact matrix $\Lambda$ by a propagator kernel and derive the execution-level analogue of the monopolist correction, with the corresponding schedule gap and phase transition.
\item \textbf{Multi-agent anticipation.} The static one-period symmetric game is now solved: \Cref{thm:multiagent} gives the Nash equilibrium interpolating monopoly and competition, \Cref{cor:erosion} the erosion of the agency gap as the field grows, and \Cref{prop:multiphase} the $N$-fold tightening of the stability threshold. What remains is the \emph{dynamic} mean-field equilibrium with transaction costs and filtering \citep{CarmonaLacker2015} --- coupling \Cref{thm:stationary} to a population of anticipating traders --- and heterogeneous impact, risk aversion, and information across agents.
\item \textbf{Beyond linear--quadratic three-way decompositions.} \Cref{thm:threeway} resolves the finite-horizon LQG case and \Cref{thm:stationary} the infinite-horizon stationary case with transaction costs and Kalman-filtered information. The remaining problem is to leave the linear--quadratic--Gaussian world entirely: nonlinear alpha dynamics, stochastic covariance, transient (propagator) and nonsymmetric or state-dependent impact, and constrained policies, where the quadratic norm and the Lyapunov trace should be replaced by an appropriate Bregman or variational-inequality gap.
\item \textbf{Finite-sample robust anticipation.} \Cref{cor:plugin,cor:hpplugin} give deterministic and high-probability edge conditions; the remaining empirical problem is to estimate nonasymptotic, asset-dependent confidence bands for $\eta_N(\delta)$ and to calibrate the ambiguity radius $\eps$ from those bands rather than choosing it as a backtest tuning knob.
\end{enumerate}

\subsection{Conclusion}
Anticipatory portfolio optimization is not a single technique but a common decision-theoretic pattern. The optimizer enriches its model by conditioning on hidden information, by using predicted return dynamics, or by internalizing the market impact of its own deployment; in each case the value of anticipation is the gap between the enriched controller and the naive estimator.

The information case gives the cleanest theoretical benchmark: value is an information-drift energy, equivalently a mutual-information or relative-entropy quantity under the usual enlargement assumptions. The forecast case shows how anticipation enters practical trading: the investor moves toward a forward-looking aim portfolio rather than a one-period Markowitz target. The impact case gives the sharpest closed form: price-taking allocation solves $(\Lambda+\gamma\Sigma)\thetana=\mu$, whereas impact-aware allocation solves $(2\Lambda+\gamma\Sigma)\thetaan=\mu$, the additional $\Lambda$ being the marginal effect of the strategy on its own opportunity set. Placing all three faces in one finite-horizon linear--quadratic--Gaussian model now makes the unification quantitative: the value is an information trace plus a deterministic inverse-precision norm, equivalently an information term, a forecast premium, an impact gap, and the unique signed forecast--impact cross term. The orthogonal projection split shows that even this signed term hides a nonnegative geometry: forecast anticipation partly cancels, reinforces, or misses the impact correction. The title is therefore literal in the LQG case, not merely interpretive. The plug-in certificates and empirical protocol add the deployment standard: do not use an enriched model just because it produces a positive backtest; use it only when the estimated edge clears the objective-error hurdle with statistical and implementation-cost margins.

The central lesson is therefore not that anticipation is always aggressive. Correctly specified anticipation can shrink positions, stabilize trading, and reduce exposure to endogenous return decay. Its danger appears only when the enriched model is estimated and then trusted too much. Robustness is the disciplined response: an anticipatory optimizer should enrich its model, but also price the uncertainty of that enrichment.

\paragraph{Status of the claims.}
The information result is classical under logarithmic utility and initial enlargement. The forecast result combines the structural linear--quadratic dynamic trading solution with an exact finite-horizon quadratic value identity. The deterministic and high-probability plug-in certificates, curvature bounds, impact, constrained-impact, phase-transition, damped-recalibration, mean--variance information-value, finite-horizon forecast-value, information--impact decomposition, full finite-horizon three-way information--forecast--impact decomposition, decision-relevant forecast norm, forecast--impact interaction bounds, orthogonal cross-term split, impact shadow-price sensitivities, the Kalman endogenization of the resolved-mean covariance, the infinite-horizon stationary value identity, the exact value of estimated anticipation with its estimation penalty and optimal reliability-ratio and Bayes shrinkage, and the $N$-trader Nash equilibrium with its competitive erosion and stability threshold are proved here in the relevant settings or operationalized where appropriate. The empirical protocol is a falsification design rather than a theorem. The transient-impact, multi-agent, nonlinear beyond local curvature, progressive-information, and finite-sample robust extensions remain open research directions.

\bibliographystyle{plainnat}
\addcontentsline{toc}{section}{References}
\bibliography{references}

\end{document}